\documentclass[a4paper,11pt]{article}

\usepackage{amsmath}
\usepackage{amssymb}
\usepackage{amsthm}
\usepackage{pstcol}
\usepackage{graphicx}
\usepackage{float}
\usepackage{url}
\usepackage{epsf}
\newcommand\N{\mathbb N}

\newcommand\Q{\mathbb Q}
\newcommand\R{\mathbb R}

\def\cT{{\mathcal T}}
\def\cR{{\mathcal R}}
\def\cF{{\mathcal F}}
\def\cP{{\mathcal P}}

\def\cC{{\mathcal C}}

\def\cN{{\mathcal N}}

\def\eref#1{(\ref{#1})}
\def\edge{ -\!\!\!- }

\usepackage[metapost, truebbox]{mfpic}
\opengraphsfile{petri}

\newtheorem{mydef}{Definition}[section]{\bfseries}{\itshape}
\newtheorem{mythm}{Theorem}[section]{\bfseries}{\itshape}
\newtheorem{myprop}{Proposition}[section]{\bfseries}{\itshape}
\newtheorem{mylem}{Lemma}[section]{\bfseries}{\itshape}
\newtheorem{mycor}{Corollary}[section]{\bfseries}{\itshape}
\newtheorem{myeg}{Example}

\title{Deficiency Zero Petri Nets and Product Form}
\author{Jean Mairesse, Hoang-Thach Nguyen\\
LIAFA, Universit\'e Paris Diderot - Paris 7,\\
75205 Paris Cedex 13\\ 
mairesse{@}liafa.jussieu.fr, ngthach{@}liafa.jussieu.fr}
\date{}
\begin{document}
\maketitle

\begin{abstract}
Consider a Markovian Petri net with race policy. The marking process
has a ``product form'' stationary
distribution if the probability of viewing a given marking can be decomposed as the product over places of terms
depending only on the local marking. First we observe that the
Deficiency Zero Theorem of Feinberg, developed for chemical reaction
networks, provides a structural and simple sufficient condition for the
existence of a product form. 
In view of this, we study the classical subclass of
free-choice nets. Roughly, we show that the only Petri nets of this class which
have
a product form are the state machines, which can
alternatively be viewed as Jackson networks. 
\end{abstract}

\section*{Introduction}\label{intro}

Queueing networks, Petri nets, and chemical reaction networks, are
three mathematical models of ``networks'', each of them with
an identified community of researchers. 

\medskip

In queueing theory, the existence of ``product form'' Markovian networks is
one of the cornerstones and jewels of the theory. 
Monographies are dedicated to the subsect, e.g. Kelly~\cite{kell79} or Van
Dijk~\cite{vandijk93}. Roughly, the
interest lies in the equilibrium behavior of a Markovian queueing
network. The existence of such an equilibrium is equivalent to the
existence of a stationary distribution $\pi$ for the queue-length
process. In some remarkable cases, $\pi$ not only exists
but has an explicit decomposable shape called ``product form''. 
The interest is two-fold. First, from a quantitative point of view, it makes the explicit computation of 
$\pi$ possible, even for large systems. Second, the product form has
important qualitative implications, like the ``Poisson-Input
Poisson-Output'' Theorems. Consequently, important and lasting efforts
have been devoted to the quest for product form queueing networks. 

\medskip

It is attractive and natural to try to develop an analog theory for
Markovian Petri nets, with the marking process replacing the
queue-length process. There is indeed a continuous string of research
on this topic since the nineties,
e.g. \cite{BoSe,CHTa96,DonSe92,FlNa91,HMSS,HLTa,lazar91}. Tools
have been developped in the process which build on classical objects of Petri
net theory (e.g. closed support T-invariants). The most
accomplished results are the ones in \cite{HMSS}. 

\medskip

In chemistry and biology, models of
  chemical reaction networks have also emerged. Such a network is
  specified by a finite set 
  of reactions between species of the type ``$2A+B \rightarrow C$'', meaning
  that two molecules of $A$ can interact with one molecule of $B$ to
  create one molecule of $C$. 
The dynamics of such models is either deterministic or stochastic, see
\cite{kurtz72}. 

Deterministic models are the most studied ones; they correspond to 
coupled sets of ordinary differential equations. 
The most significant result is arguably the Deficiency
Zero Theorem of Feinberg~\cite{feinberg79}. Deficiency Zero is a structural property
which can be very easily checked knowing the shape of the reactions in
a chemical network. It does not refer to any assumption on the
associated dynamics. Feinberg Theorem states that if a network
satisfies the Deficiency Zero condition, then the associated
deterministic dynamic model has remarkable stability properties. An
intermediate result is to prove that a set of non-linear equations
(NLE) have a strictly positive solution. 

Stochastic models of chemical reaction networks
correspond to continuous-time Markov processes of a specific shape. 
Such models were considered in Chapter 8 of the seminal book by
Kelly~\cite{kell79}. There it is proved that if a set of non-linear
``traffic equations'' (NLTE) have a strictly positive solution then 
the Markov process has a product form. 

\medskip

How does Feinberg result connect with product form Markovian Petri
nets~? 

\medskip

A first observation is that chemical reaction networks and Petri nets are two different
descriptions of the same object. This has been identified by different
authors in the biochemical community: see for instance \cite{ALSo} and the
references therein. Conversely, Petri nets were originally introduced
by Carl Adam Petri to model chemical processes, see \cite{PeRe}. 

A second observation was made recently by Anderson, Craciun, and
Kurtz~\cite{ACKu}. They observe that the NLE of Feinberg and the NLTE of
Kelly are the same. It implies that if a chemical network has
deficiency zero, then the stochastic dynamic model has a
product form. 

In the present paper, we combine the two observations. The
Deficiency Zero condition provides a sufficient condition for a
Markovian Petri net to have a product form. The Deficiency Zero condition is
equivalent to another criterion known in the Petri net
literature~\cite{HMSS}. The advantage is that deficiency is easy to
compute and handle. 

\medskip

The class of Petri nets whose Markovian version have a product form is an
interesting one. It is therefore natural to study how this class 
intersects with the classical families of Petri nets: state machines
and free-choice Petri nets. We use the simplicity of the Deficiency
Zero condition to carry out this study. 

The central result that we prove is, in a sense, a negative result. 
We show that within the class of free-choice Petri
nets, the only ones which have a product form are closely related to state
machines. We also show that the Markovian state
machines are ``equivalent'' to Jackson networks. The latter form the
most basic and classical example of product form queueing networks. 

\medskip

A conference version of the present paper appeared in
\cite{MaNg09}. Compared with \cite{MaNg09}, additional results
have been proved: the equivalence between deficiency 0
and the condition of Haddad $\&$ al \cite{HMSS}
(Prop. \ref{pr:HMSS2}), and the results in Section \ref{sse-addi}.

\section{Model}\label{sec:model}

We use the notation $\R^* = \R - \{0\}$. 
The coordinate-wise ordering of $\R^k$ is denoted by the symbol
$\leqslant$. We say that $x\in \R^k$ is {\em strictly positive} if
$x_i>0$ for all $i$. We denote by ${\bf 1}_{S}$ the indicator
function of $S$, that is the mapping taking value 1 inside $S$ and 0
outside. 

\subsection{Petri nets}\label{subsec:PN}

Our definition of Petri net is standard, with weights on the arcs. 

\begin{mydef}[Petri net]\label{def:PN}
A {\em Petri net} is a 6-tuple $(\mathcal P, \mathcal T, \mathcal F, I, O,
M_0)$ where:
\begin{itemize}
\item $(\mathcal P, \mathcal T, \mathcal F)$ is a directed bipartite graph, that
is, $\mathcal P$ and $\mathcal T$ are non-empty and finite disjoint
sets, and $\mathcal
F$ is a subset of $(\mathcal P \times \mathcal T) \cup (\mathcal T \times \mathcal
P)\, ;$

\item $I: \mathcal T \rightarrow \mathbb N^{\mathcal P}$, $O: \mathcal T
\rightarrow \mathbb N^{\mathcal P}$ are such that $[ I(t)_p > 0 \Leftrightarrow
(p, t) \in \mathcal F ] $ and $[ O(t)_p > 0 \Leftrightarrow (t, p) \in
  \mathcal F ] \,;$

\item $M_0$ belongs to $\mathbb N^{\mathcal P}.$
\end{itemize}
\end{mydef}

The elements of $\mathcal P$ are called {\em places}, those of $\mathcal T$ are
called {\em transitions}. The 5-tuple $(\mathcal P, \mathcal T,
\mathcal F,I,O)$ is called the {\em Petri graph}. 
The vectors $I(t)$ and $O(t)$, $t \in \mathcal T$, are called
the {\em input bag} and the {\em output bag} of the
transition $t$. The {\em weight} of the arc $(p, t) \in \mathcal F$
(resp. $(t, p)$) is $I(t)_p$ (resp. $O(t)_p$). 

An element of $\mathbb N^{\mathcal P}$ is called a
{\em marking}, and $M_0$ is called the {\em initial marking}. 

\medskip

Petri nets inherit the usual terminology of graph theory. 
Graphically, a Petri net is represented by a directed graph in which places are represented
by circles and transitions by rectangles. The initial marking is also
materialized: if $M_0(p) = k$, then $k$ tokens are drawn inside the
circle $p$. By convention, the weights different from 1 are
represented on the arcs. See Figures \ref{fig:LTEnotNLTE} or \ref{fi-4} for examples. 

\medskip

A Petri net is a dynamic object. The Petri graph always remains unchanged,
but the marking evolves according to the {\em firing rule}. 
A transition $t$ is {\em enabled} in the marking $M$ if 
$M \geq I(t)$, then  $t$ may {\em fire} which transforms the marking
from $M$ into 
\[
M'= M - I(t) + O(t)\:.
\]
We write $M \xrightarrow{t} M'$. A marking $M'$ is {\em reachable}
from a marking $M$ if there exists a (possibly empty) sequence of
transitions $t_1,...,t_k$ and sequence of markings
$M_1,...,M_{k-1},$ such that 
\[
M \xrightarrow{t_1} M_1
\xrightarrow{t_2} \cdots \xrightarrow{t_{k-1}} M_{k-1}
\xrightarrow{t_k} M'\:.
\]
We
denote by $\mathcal R(M)$ the set of markings which are reachable from $M$.

\begin{mydef}[Marking graph]\label{def:markinggraph}
The {\em marking graph} of a Petri net with initial marking $M_0$
is the directed graph with 
\begin{itemize}
\item nodes: $\mathcal R(M_0)\,,$ arcs: $M \rightarrow M' \mbox{ if
}\, \exists t \in \cT, \ M \xrightarrow{t} M'.$
\end{itemize}
\end{mydef}

The marking graph defines the state space on which the marking may
evolve. Observe that the marking graph may be finite or infinite. 
In Section \ref{subsec:SPN}, we will define a {\em Markovian Petri
  net} as a continuous-time Markovian process evolving on the marking
graph. 

\medskip

The analysis of Petri nets relies heavily on linear algebra techniques, the
central object being the incidence matrix. 

The {\em incidence matrix} $N$ of the Petri net $(\mathcal P, \mathcal T,
\mathcal F, I, O, M_0)$ is the  $(\mathcal P \times \mathcal
T)$-matrix $N$ defined by: 
\begin{equation}\label{eq-def-incidence}
N_{s,t} = O(t)_s - I(t)_s\,.
\end{equation}

\begin{myeg}
Figure \ref{fig:LTEnotNLTE} represents a Petri net with
places $\{p_1, p_2, p_3, p_4\}$ and transitions $\{t_1, t_2, t_3,
t_4, t_5, t_6\}$.  The initial marking is $M_0=(2, 0, 0, 1).$ 
Here all the weights of the arcs are equal to 1. 

\begin{figure}[H]
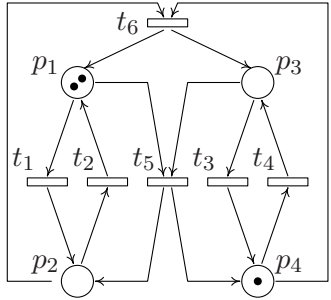

\centering
\begin{mfpic}{-75}{75}{-60}{75}
\tlabelsep{3pt}
\circle{(-33.75, -37.5), 6}
\tlabel[br](-37.5, -37.5){$p_2$}
\circle{(33.75, -37.5), 6}
\tlabel[bl](37.5, -37.5){$p_4$}
\circle{(-33.75, 37.5), 6}
\tlabel[br](-37.5, 37.5){$p_1$}
\circle{(33.75, 37.5), 6}
\tlabel[bl](37.5, 37.5){$p_3$}

\rect{(-7.5, 1.5), (7.5, -1.5)}
\tlabel[br](-1.5, 1.5){$t_5$}
\rect{(-7.5, 61.5), (7.5, 58.5)}
\tlabel[cr](-7.5, 60){$t_6$}
\rect{(-30, -1.5), (-15, 1.5)}
\tlabel[br](-24, 1.5){$t_2$}
\rect{(-52.5, -1.5), (-37.5, 1.5)}
\tlabel[br](-46.5, 1.5){$t_1$}
\rect{(15, -1.5), (30, 1.5)}
\tlabel[br](21, 1.5){$t_3$}
\rect{(37.5, -1.5), (52.5, 1.5)}
\tlabel[br](43.5, 1.5){$t_4$}

\arrow \polyline{(-43.5, -37.5), (-60, -37.5), (-60, 67.5), (-3.75, 67.5),
(-1.5, 62.25)}
\arrow \polyline{(40.5, -37.5), (60, -37.5), (60, 67.5), (3.75, 67.5), (1.5,
62.25)}
\arrow \polyline{(-1.5, 57.25), (-31, 43)}
\arrow \polyline{(1.5, 57.25), (31, 43)}

\arrow \polyline{(-27, 37.5), (-7.5, 37.5), (-1.5, 2.25)}
\arrow \polyline{(27, 37.5), (7.5, 37.5), (1.5, 2.25)}
\arrow \polyline{(-1.5, -2.25), (-7.5, -37.5), (-27, -37.5)}
\arrow \polyline{(1.5, -2.25), (7.5, -37.5), (27, -37.5)}

\arrow \polyline{(-35.25, 30.75), (-45, 2.25)}
\arrow \polyline{(-45, -2.25), (-35.25, -30.75)}
\arrow \polyline{(-32.25, -30.75), (-22.5, -2.25)}
\arrow \polyline{(-22.5, 2.25), (-32.25, 30.75)}

\arrow \polyline{(35.25, -30.75), (45, -2.25)}
\arrow \polyline{(45, 2.25), (35.25, 30.75)}
\arrow \polyline{(32.25, 30.75), (22.5, 2.25)}
\arrow \polyline{(22.5, -2.25), (32.25, -30.75)}

\point[3pt]{(-32, 39), (-35, 36), (33.75, -37.5)}

\end{mfpic}
\caption{Petri net.}
\label{fig:LTEnotNLTE}
\end{figure}

\end{myeg}

\subsection{Deficiency and weak reversibility}\label{subsec:chem}

Consider a Petri net $(\mathcal P, \mathcal T, \mathcal F, I, O,
M_0)$. From now on, we assume that no two transitions have the same
pair input bag - output bag. 
So we can identify each $t$ with the ordered pair
$(I(t), O(t))$. 

In the context of Markovian Petri nets with race policy (see \S
\ref{subsec:SPN}), 
this assumpion is made without loss of generality. Indeed, several
transitions with the same pair of bags can be replaced by a
single transition whose transition rate is the sum of the original
rates. Also this transformation does not modify the deficiency or weak
reversibility (to be defined below). 

\medskip

Under this assumption, the Petri net can  be viewed as a
triple $(\mathcal P, \mathcal T \subset
\N^{\mathcal P} \times \N^{\mathcal P}, M_0 \in \N^{\mathcal P})$. (In
particular, the flow relation $\cF$ is encoded in $\cT$.) 

Petri nets have appeared with this presentation in different contexts and under
different names: 
{\em vector addition systems} (see for instance \cite{reut90}), or {\em chemical reaction
networks} (see for instance \cite{feinberg79,ACKu}). 

\medskip

In the chemical context, the elements of $\mathcal P$ are 
species. The marking is the number of molecules of the different
species. The elements of $\mathcal T$ are {\em reactions}. 
A reaction
$(c, d) \in \mathbb N^{\mathcal P} \times \mathbb N^{\mathcal P}$ is represented
as follows:
\[\sum_{p \in \mathcal P} c_pp \longrightarrow \sum_{p \in \mathcal P} d_pp.\]

\begin{myeg}
The ``chemical'' form of the Petri net in Figure \ref{fig:LTEnotNLTE} is:
\begin{equation}\label{eq-chem}
p_1 \rightleftarrows  p_2 \qquad p_3 \rightleftarrows  p_4 \qquad p_1
+ p_3  \rightleftarrows  p_2 + p_4 \:.
\end{equation}
\end{myeg}

Let us now introduce two notions, deficiency and weak reversibility, which are borrowed from the
chemical literature.

\begin{mydef}[Reaction graph]\label{def:reactiongraph}
Let $(\mathcal P, \mathcal T\subset
\N^{\mathcal P} \times \N^{\mathcal P},  M)$ be a Petri net. 
A {\em complex} is a vector $u$ in $\N^{\cP}$ such that: $\exists v
\in \N^\cP, \ (u,v)\in \cT \mbox{ or }  (v,u)\in \cT$. The set of all
complexes is denoted by $\mathcal C$. The {\em reaction graph}
associated to the Petri net is the directed
graph with nodes: $\cC$, arcs: $u \rightarrow v$ if $(u,v)\in \cT$. 
\end{mydef}

Let $A$ be the {\em node-arc incidence matrix} of
the reaction graph, that is the $(\cC\times\cT)$-matrix defined by
\begin{equation}\label{eq-A}
A_{u,t} = - {\bf 1}_{\{I(t)=u\}} + {\bf 1}_{\{O(t)=u\}}\:.
\end{equation}

\begin{mylem}\label{lem:rankA}
Consider a Petri net with set of complexes $\cC$. Let $\ell$ be the number of
connected components of the reaction graph. The rank of the node-arc
incidence matrix satisfies:
  \begin{equation}
    \mbox{rank}(A) = |\cC| - \ell\,.
    \label{eq:rankA}
  \end{equation}
\end{mylem}

\begin{proof}
 Assume that $\ell =1$. Consider $x\in \R^{\cC}-\{(0,\dots, 0)\}$
such that $xA=(0,\dots, 0)$. Let $C$ be such that $x_C\neq
0$. Consider $D\in \cC$. Since $\ell =1$, there exists an undirected path
$(C=C_0)\edge C_1 \edge \cdots \edge (C_k=D)$ in
the reaction graph. Assume wlog that $C_i \rightarrow C_{i+1} $ and let $t_i\in \cT$ be such that $I(t_i)=C_i,
O(t_i)=C_{i+1}$. By definition of $A$, we have $(xA)_{t_i} =
x_{C_{i+1}} -x_{C_i}$. So we have $x_{C_i}=x_{C_{i+1}}$ for all $i$,
and $x_C=x_D$. We have proved that 
$xA=(0,\dots,0) \implies x \in \R (1,\dots ,1)$. 
Conversely, by definition of $A$, we have $(1,\dots ,1)A =
(0,\dots,0)$. Hence $\mbox{dim} \ ker(A) = 1$, and by the rank theorem, $\mbox{rank}(A)= |\cC| -1$. 
For a general value of $\ell$,
we get similarly that $\mbox{rank}(A)= |\cC| -\ell$.
\end{proof}

A central notion in what follows is the {\em deficiency} of a Petri net. 

\begin{mydef} [Deficiency]\label{def:deficiency}
With the above notations, the {\em deficiency} of the Petri net is
\[\delta =  |\mathcal C|  - \ell - \mbox{rank}(N) = \mbox{rank}(A) - \mbox{rank}(N)\:.\]
\end{mydef}

Of particular importance are the Petri nets
with {\bf deficiency 0}. This class will be central in the study 
of Markovian Petri nets having a product form, see 
Section \ref{sec:prodres}. Such Petri nets are ``extremal'', in a
sense made precise by the next proposition.

\begin{myprop}\label{pr-nonneg}
For all $x \in \R^{\cT}$: 
\begin{equation}\label{eq:eliminationAN2}
\bigl[ Ax = 0 \bigr] \implies \bigl[ Nx
  = 0 \bigr] \:.
\end{equation}
In particular it implies that: $\mbox{rank}(A) \geq
\mbox{rank}(N)$. Equivalently, the deficiency of a Petri net is always
greater or equal to 0. 
\end{myprop}

The non-negativity of the deficiency appears in Feinberg~\cite{feinberg79}, using a different
argument. 

\begin{proof}
Using the definition of the matrix $A$, we can rewrite the condition
$Ax = 0$ as:
\begin{equation}
\forall c \in \cC, \qquad \sum_{t: O(t) = c} x_t - \sum_{t: I(t) = c}
x_t = 0\,, \label{eq:Ax2}
\end{equation}
Now for each place $s$, we have:
\begin{eqnarray*}
\sum_{t \in \cT}N_{s,t} x_t &=& \sum_{t \in \cT}[O(t)_s - I(t)_s] x_t \\
&=&  \sum_{c \in \cC} \bigl( \sum_{t: O(t) = c} c_s x_t - \sum_{t: I(t) = c} c_s
x_t \bigr)  \ = \ 
\sum_{c \in \cC} c_s \bigl( \sum_{t: O(t) = c} x_t - \sum_{t: I(t) = c}
x_t \bigr)\,. \\
\end{eqnarray*}
Using \eref{eq:Ax2}, this last sum is equal to $0$.

The inequality $\mbox{rank}(A) \geq \mbox{rank}(N)$ is equivalent to
$\mbox{dim ker}(A) \leq \mbox{dim ker}(N)$, which follows immediately from
\eref{eq:eliminationAN2}.
\end{proof}

The second central notion is {\em weak reversibility}. 

\begin{mydef}[Weak reversibility]\label{def:weakrevers}
A Petri net is {\em weakly reversible (WR)} if every connected component of the
reaction
graph is strongly connected.
\end{mydef}

Weak reversibility is a restrictive property, see Section
\ref{subsec:freechoice}. 
It is important to observe that a connected and weakly reversible
Petri net is not necessarily strongly connected. An example is given
below. 

\begin{figure}[H]
\centering
\begin{mfpic}{-95}{95}{-20}{50}
\shiftpath{(-30, 0)}\circle{(-15, 0), 6}
\shiftpath{(-30, 0)}\rect{(7.5, 7.5), (10.5, -7.5)}
\shiftpath{(-30, 0)}\rect{(-40.5, 7.5), (-37.5, -7.5)}

\arrow\shiftpath{(-30, 0)}\polyline{(-36, 0), (-22.5, 0)}
\arrow\shiftpath{(-30, 0)}\polyline{(-7.5, 0), (6, 0)}

\tlabelsep{3pt}
\tlabel[bc](-45, 7.5){$p$}

\point[3pt]{(67.5, 0), (37.5, 0)}
\arrow\curve{(38.25, .75), (52.5, 4.5), (66.75, .75)}
\arrow\curve{(66.75, -.75), (52.5, -4.5), (38.25, -.75)}

\tlabel[br](36.75, .75){$\emptyset$}
\tlabel[bl](68.25, .75){$p$}

\end{mfpic}
\end{figure}

Elementary circuits of the
reaction graph can be identified with the so-called ``minimal closed
support T-invariants'' of the Petri net literature
(see \cite{BoSe}). 
In particular, a Petri net is weakly reversible if and only if it is covered by minimal closed
support T-invariants. Such Petri nets are called {\em $\Pi$-nets} in \cite{HMSS}.

\paragraph{Weak reversibility and deficiency 0.}

Weak reversibility and deficiency 0 are two independent
properties, see Figure \ref{fi-4}. The upper-left Petri net is an instance of the famous
``dining philosopher'' model. The Petri nets in Fig. \ref{fi-4} are
live and bounded except for the upper-right one.

\begin{figure}[H]
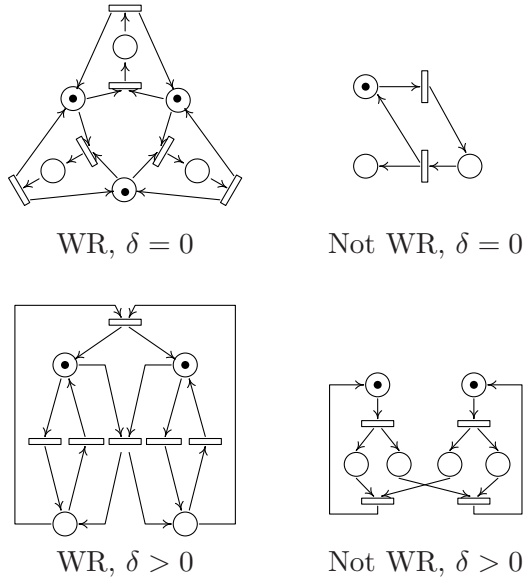

\centering
\begin{mfpic}{-150}{150}{-112.5}{112.5}
\tlabelsep{3pt}

\shiftpath{(-75, 90)}\circle{(0,0), 4.5}
\arrow\shiftpath{(-75, 90)}\polyline{(0, 6), (0, 12.75)}
\arrow\shiftpath{(-75, 90)}\polyline{(0, -12.75), (0, -6)}
\shiftpath{(-75, 90)}\rect{(-6, 13.5), (6, 15.75)}
\shiftpath{(-75, 90)}\rect{(-6, -15.75), (6, -13.5)}
\arrow\shiftpath{(-75, 90)}\polyline{(-5.25, 12.75), (-18, -15.75)}
\arrow\shiftpath{(-75, 90)}\polyline{(5.25, 12.75), (18, -15.75)}
\arrow\shiftpath{(-75, 90)}\polyline{(-14.25, -19.5), (-1.5, -16.5)}
\arrow\shiftpath{(-75, 90)}\polyline{(14.25, -19.5), (1.5, -16.5)}
\shiftpath{(-75, 90)}\circle{(-19.5, -19.5), 4.5}

\shiftpath{(-75, 90)}\shiftpath{(-27, -46.764)}\rotatepath{(0, 0),
120}\circle{(0,0), 4.5}
\arrow\shiftpath{(-75, 90)}\shiftpath{(-27, -46.764)}\rotatepath{(0, 0),
120}\polyline{(0, 6), (0, 12.75)}
\arrow\shiftpath{(-75, 90)}\shiftpath{(-27, -46.764)}\rotatepath{(0, 0),
120}\polyline{(0, -12.75), (0, -6)}
\shiftpath{(-75, 90)}\shiftpath{(-27, -46.764)}\rotatepath{(0, 0),
120}\rect{(-6, 13.5), (6, 15.75)}
\shiftpath{(-75, 90)}\shiftpath{(-27, -46.764)}\rotatepath{(0, 0),
120}\rect{(-6, -15.75), (6, -13.5)}
\arrow\shiftpath{(-75, 90)}\shiftpath{(-27, -46.764)}\rotatepath{(0, 0),
120}\polyline{(-5.25, 12.75), (-18, -15.75)}
\arrow\shiftpath{(-75, 90)}\shiftpath{(-27, -46.764)}\rotatepath{(0, 0),
120}\polyline{(5.25, 12.75), (18,
-15.75)}
\arrow\shiftpath{(-75, 90)}\shiftpath{(-27, -46.764)}\rotatepath{(0, 0),
120}\polyline{(-14.25, -19.5),
(-1.5, -16.5)}
\arrow\shiftpath{(-75, 90)}\shiftpath{(-27, -46.764)}\rotatepath{(0, 0),
120}\polyline{(14.25, -19.5), (1.5, -16.5)}
\shiftpath{(-75, 90)}\shiftpath{(-27, -46.764)}\rotatepath{(0, 0),
120}\circle{(-19.5, -19.5), 4.5}

\shiftpath{(-75, 90)}\shiftpath{(27, -46.764)}\rotatepath{(0, 0),
-120}\circle{(0,0), 4.5}
\arrow\shiftpath{(-75, 90)}\shiftpath{(27, -46.764)}\rotatepath{(0, 0),
-120}\polyline{(0, 6), (0, 12.75)}
\arrow\shiftpath{(-75, 90)}\shiftpath{(27, -46.764)}\rotatepath{(0, 0),
-120}\polyline{(0, -12.75), (0, -6)}
\shiftpath{(-75, 90)}\shiftpath{(27, -46.764)}\rotatepath{(0, 0),
-120}\rect{(-6, 13.5), (6, 15.75)}
\shiftpath{(-75, 90)}\shiftpath{(27, -46.764)}\rotatepath{(0, 0),
-120}\rect{(-6, -15.75), (6, -13.5)}
\arrow\shiftpath{(-75, 90)}\shiftpath{(27, -46.764)}\rotatepath{(0, 0),
-120}\polyline{(-5.25, 12.75), (-18, -15.75)}
\arrow\shiftpath{(-75, 90)}\shiftpath{(27, -46.764)}\rotatepath{(0, 0),
-120}\polyline{(5.25, 12.75), (18,
-15.75)}
\arrow\shiftpath{(-75, 90)}\shiftpath{(27, -46.764)}\rotatepath{(0, 0),
-120}\polyline{(-14.25, -19.5),
(-1.5, -16.5)}
\arrow\shiftpath{(-75, 90)}\shiftpath{(27, -46.764)}\rotatepath{(0, 0),
-120}\polyline{(14.25, -19.5), (1.5, -16.5)}
\shiftpath{(-75, 90)}\shiftpath{(27, -46.764)}\rotatepath{(0, 0),
-120}\circle{(-19.5, -19.5), 4.5}


\point[3pt]{(-94.5, 70.5)}
\point[3pt]{(-55.5, 70.5)}
\point[3pt]{(-75, 35.25)}

\tlabel[cc](-75, 15){WR, $\delta = 0$}

\shiftpath{(15, 75)}\circle{(0,0), 4.5}
\arrow\shiftpath{(15, 75)}\polyline {(5.25, 0), (20.25, 0)}
\shiftpath{(15, 75)}\rect{(21, 6), (23.25, -6)}
\arrow\shiftpath{(15, 75)}\polyline{(24, 0), (37.5, -24.75 )}
\shiftpath{(15, 75)}\circle{(39, -30), 4.5}
\arrow\shiftpath{(15, 75)}\polyline{(33.75, -30), (24, -30)}
\shiftpath{(15, 75)}\rect{(21, -24), (23.25, -36)}
\arrow\shiftpath{(15, 75)}\polyline{(20.25, -28.5), (4.5, -3)}
\arrow\shiftpath{(15, 75)}\polyline{(20.25, -30), (5.25, -30)}
\shiftpath{(15, 75)}\circle{(0, -30), 4.5}

\point[3pt]{(15, 75)}

\tlabel[cc](37.5, 15){Not WR, $\delta = 0$}

\shiftpath{(-75, -60)}\circle{(-22.5, -30), 4.5}
\shiftpath{(-75, -60)}\circle{(22.5, -30), 4.5}
\shiftpath{(-75, -60)}\circle{(-22.5, 30), 4.5}
\shiftpath{(-75, -60)}\circle{(22.5, 30), 4.5}

\shiftpath{(-75, -60)}\rect{(-6, 2.25), (6, 0)}
\shiftpath{(-75, -60)}\rect{(-6, 47.25), (6, 45)}
\shiftpath{(-75, -60)}\rect{(-21, 2.25), (-8, 0)}
\shiftpath{(-75, -60)}\rect{(-24, 2.25), (-36, 0)}
\shiftpath{(-75, -60)}\rect{(8, 2.25), (21, 0)}
\shiftpath{(-75, -60)}\rect{(24, 2.25), (36, 0)}

\arrow \shiftpath{(-75, -60)}\polyline{(-27.75, -30), (-41.25, -30), (-41.25,
52.5), (-3.75, 52.5), (-1.5, 48)}
\arrow \shiftpath{(-75, -60)}\polyline{(27.75, -30), (41.25, -30), (41.25,
52.5), (3.75, 52.5),
(1.5, 48)}
\arrow \shiftpath{(-75, -60)}\polyline{(-1.5, 44.25), (-18.75, 32.25)}
\arrow \shiftpath{(-75, -60)}\polyline{(1.5, 44.25), (18.75, 32.25)}

\arrow \shiftpath{(-75, -60)}\polyline{(-17.25, 30), (-7.5, 30), (-1.5, 3)}
\arrow \shiftpath{(-75, -60)}\polyline{(17.25, 30), (7.5, 30), (1.5, 3)}
\arrow \shiftpath{(-75, -60)}\polyline{(-1.5, -3), (-7.5, -30), (-17.25, -30)}
\arrow \shiftpath{(-75, -60)}\polyline{(1.5, -3), (7.5, -30), (17.25, -30)}

\arrow \shiftpath{(-75, -60)}\polyline{(-24, 24.75), (-30, 3)}
\arrow \shiftpath{(-75, -60)}\polyline{(-30, -0.75), (-24, -24.75)}
\arrow \shiftpath{(-75, -60)}\polyline{(-21, -24.75), (-15, -.75)}
\arrow \shiftpath{(-75, -60)}\polyline{(-15, 3), (-21, 24.75)}

\arrow \shiftpath{(-75, -60)}\polyline{(24, -24.75), (30, -.75)}
\arrow \shiftpath{(-75, -60)}\polyline{(30, 3), (24, 24.75)}
\arrow \shiftpath{(-75, -60)}\polyline{(21, 24.75), (15, 3)}
\arrow \shiftpath{(-75, -60)}\polyline{(15, -.75), (21, -24.75)}

\point[3pt]{(-97.5, -30), (-52.5, -30)}
\tlabel[cc](-75, -105){WR, $\delta > 0$}

\shiftpath{(19.5, -37.5)}\circle{(0,0), 4.5}
\arrow\shiftpath{(19.5, -37.5)}\polyline{(0, -5.25), (0, -12.75)}
\shiftpath{(19.5, -37.5)}\rect{(-6, -13.5), (6, -15.75)}
\arrow\shiftpath{(19.5, -37.5)}\polyline{(-1.5, -16.5), (-8, -24.75)}
\shiftpath{(19.5, -37.5)}\circle{(-8, -30), 4.5}
\arrow\shiftpath{(19.5, -37.5)}\polyline{(1.5, -16.5), (8, -24.75)}
\shiftpath{(19.5, -37.5)}\circle{(8, -30), 4.5}

\shiftpath{(19.5, -37.5)}\circle{(36,0), 4.5}
\arrow\shiftpath{(19.5, -37.5)}\polyline{(36, -5.25), (36, -12.75)}
\shiftpath{(19.5, -37.5)}\rect{(30, -13.5), (42, -15.75)}
\arrow\shiftpath{(19.5, -37.5)}\polyline{(34.5, -16.5), (27, -24.75)}
\shiftpath{(19.5, -37.5)}\circle{(27, -30), 4.5}
\arrow\shiftpath{(19.5, -37.5)}\polyline{(37.5, -16.5), (45, -24.75)}
\shiftpath{(19.5, -37.5)}\circle{(45, -30), 4.5}

\shiftpath{(19.5, -37.5)}\rect{(-6, -42.75), (6, -45)}
\arrow\shiftpath{(19.5, -37.5)}\polyline{(-8, -35.25), (-1.5, -42)}
\arrow\shiftpath{(19.5, -37.5)}\polyline{(27, -35.25), (1.5, -42)}
\arrow\shiftpath{(19.5, -37.5)}\polyline{(0, -45.75), (0, -48.75), (-18,
-48.75), (-18, 0),
(-5.25,
0)}

\shiftpath{(19.5, -37.5)}\rect{(30, -42.75), (42, -45)}
\arrow\shiftpath{(19.5, -37.5)}\polyline{(8, -35.25), (34.5, -42)}
\arrow\shiftpath{(19.5, -37.5)}\polyline{(45, -35.25), (37.5, -42)}
\arrow\shiftpath{(19.5, -37.5)}\polyline{(36, -45.75), (36, -48.75), (54,
-48.75), (54, 0),
(41.25,
0)}

\point[3pt]{(19.5, -37.5), (55.5, -37.5)}
\tlabel[cc](37.5, -105){Not WR, $\delta > 0$}

\end{mfpic}

\caption{Deficiency zero and weak reversibility are independent.}
\label{fi-4}
\end{figure}

\paragraph{Algorithmic complexity.}

Weak reversibility and deficiency 0 are algorithmically simple to
check. Let us determine the time complexity 
of the algorithms with respect to the size of the Petri net (number of places
and transitions). 

Observe first that the number of complexes is bounded by $2\cT$. 
Building the reaction graph from the Petri graph 
can be done in time $\mathcal{O}(\cP\cT^2)$. 
A depth-first-search algorithm on the reaction graph enables to check
the weak reversibility and to compute the number of connected
components ($\ell$). The DFS algorithm runs in time
$\mathcal{O}(\cP\cT)$.
Computing the rank of the incidence matrix can be done in
time $\mathcal{O}(\cP \cT^2)$ using a Gaussian elimination. 

Globally the complexity is $\mathcal{O}(\cP \cT^2)$ for computing
the deficiency as well as for checking weak-reversibility. 

\medskip

\subsection{Markovian Petri nets with race policy}\label{subsec:SPN}

A Petri net is a logical object with no physical time 
involved. There exist several alternative ways to define timed models of Petri
nets, see for instance \cite{ABBCCC,BCOQ}. We consider the model of
Markovian Petri nets with race policy. 
The rough description
is as follows. 

\medskip

With each enabled transition is associated 
a ``countdown
clock'' whose positive initial value is set at random. When a clock
reaches 0, the corresponding transition fires. This changes the set
of enabled transitions and all the clocks get reinitialized. 
The initial values of the clocks are chosen independently, according to an
exponential distribution whose rate depends on the transition and on
the current marking. With 
probability 1, no two clocks reach zero at the same time so the model
is  unambiguously defined.  
Enabled transitions are involved in a ``race'': the transition to fire
is the one whose 
clock will reach zero first. 

\medskip

We now proceed to a formal definition of the model. 

When $I(t)=O(t)$, the firing of transition $t$ does not modify the
marking, nor the distribution of the value of the clocks (memoryless
property of the exponential). Without loss of
generality, we assume from now on that $I(t)\neq O(t)$ 
for all $t$. 

\begin{mydef}[Markovian Petri net with race policy]\label{def:SPN}
A {\em Markovian Petri net (with race policy)} is  formed by a Petri
net $(\cP,\cT,\cF,I,O,M_0)$ and a set
of {\em rate functions} $(\mu_t)_{t\in \cT}$, $\mu_t: \mathcal R(M_0)
\rightarrow \mathbb R_+^*$, satisfying 
\begin{equation}\label{eq:generalratefunction}
\mu_t(M) = \begin{cases}
\kappa_t \Psi(M -I(t))\Phi(M)  & \mbox{if } M
\geqslant I(t) \\
0 & \mbox{otherwise}
\end{cases}\:,
\end{equation}
for some constants $\kappa_t\in \R_+^*, \ t\in \cT$, and some functions
$\Psi$ and $\Phi$ valued in $\R_+^*$. 
The marking
evolves as a continuous-time jump Markov process with state space
$\mathcal R(M_0)$ and infinitesimal generator $Q=(q_{M,M'})_{M,M'}$,
given by:
\begin{equation}\label{eq:Q}
\forall M, \ \forall M' \neq M, \qquad q_{M, M'} = \sum_{t: M
  \xrightarrow{t} M'} 
\mu_t(M), \qquad \forall M, \quad q_{M,M} = -\sum_{M'\neq M} q_{M,M'}\,. 
\end{equation}
\end{mydef}

The shape \eref{eq:Q} for the infinitesimal generator is the
transcription of the informal description given at the beginning of
the section. 

\medskip

The condition \eref{eq:generalratefunction} for the rate functions
$(\mu_t)_{t\in \cT}$ is the
same as the one in \cite{HLTa} and \cite[Section 2]{HMSS}. (In \cite[Section
  3]{HMSS}, an even more general shape for the rate function is considered.)
Condition \eref{eq:generalratefunction} is specifically cooked up in order for 
the product form result of Theorem \ref{thm:Kelly} to hold, which
explains its artificial shape. 
This general condition englobes two classical types of rate functions: the constant rates
and the mass-action rates. 

\medskip

\textbf{Constant rates.} In the Petri net literature, the
standard assumption is that the firing rates are constant: 
\begin{equation}\label{eq-constant}
\exists \kappa_t  \in \R_+^*, \ \forall M \in \mathcal R(M_0),
I(t)\leqslant M, \qquad  \mu_t(M)= 
\kappa_t \:.
\end{equation}

\medskip

\textbf{Mass-action rates.} In the chemical literature, the rate is
often proportional to the
number of different subsets of tokens (i.e. molecules) that can be
involved in the firing (i.e. reaction). 
More precisely:
\begin{equation}\label{eq-massaction}
\forall M \in \mathcal R(M_0), I(t)\leqslant M, \qquad  \mu_t(M) = \kappa_t \prod_{p: I(t)_p \neq 0} \frac{M_p!}{(M_p -
I(t)_p)!} \,.
\end{equation}
Such rates are said to be of {\em mass-action} form and the
corresponding stochastic process has {\em mass-action kinetics}. To
obtain \eref{eq-massaction} from \eref{eq:generalratefunction}, set 
$\Phi, \Psi^{-1}: \N^{\cP} \rightarrow \R^*_+, x \mapsto \prod_p x_p !$.

\section{Product form results}\label{sec:prodres}

We are interested in the equilibrium behavior of Markovian Petri
nets.
This section presents the product form results which exist in the
literature. We gather results which were spread out, obtained independently either
in the Petri net community, or in the chemical one. 

\medskip

Let $Q$ be the infinitesimal generator of the marking process. 
An invariant measure $\pi$ of the process
is characterized by the balance equations $\pi Q =0$, that is: 
$\forall x \in \mathcal R(M_0)$, 
\begin{equation}
\pi(x) \sum_{t: x \geqslant I(t)} \mu_t(x) = \sum_{t: x \geqslant O(t)}\pi(x + I(t) -
O(t)) \mu_t(x + I(t) - O(t))\,. \label{eq:invariantmeasure}
\end{equation}
A stationary distribution is an invariant probability measure. It is
characterized by $\pi Q =0$, $\sum_x \pi(x)=1$. 
If $\pi$ is an invariant measure and $K=\sum_{x \in \mathcal R(M_0)} \pi(x) < +\infty$, then
$\pi/K= (\pi(x)/K)_x$ is a stationary distribution. 

\medskip

When the marking graph is strongly connected, the marking process is
irreducible. It follows from the basic Markovian theory that the
stationary distribution is unique when it exists (the ergodic case). 
When the state space is finite, irreducibility implies ergodicity. 

\subsection{Non-linear traffic equations and Kelly's Theorem}

\begin{mydef}[Non-linear traffic equations]\label{def:NLTE}
Consider a Markovian Petri net with general rates. Let $\cC$ be the set of complexes. 
We call {\em non-linear traffic equations (NLTE)} the
equations over the unknowns $(x_p)_{p\in \cP}$ defined by: $\forall C \in \cC$, 
\begin{equation}
\prod_{p: C_p\neq 0} x_p^{C_p} \sum_{t: I(t) = C}\kappa_t  = \sum_{t:
O(t) = C}\kappa_{t} \prod_{p: I(t)_p\neq 0}x_p^{I(t)_p}\,. \label{eq:NLTE}
\end{equation}
(With the usual convention that the empty product is 1.)
\end{mydef}

The NLTE can be viewed as a kind of balance equations (what goes in
equals what goes out) at the level of complexes. Their central role 
appears in 
the next theorem which is essentially due to Kelly~\cite[Theorem
  8.1]{kell79} (see also \cite[Theorem 4.1]{ACKu}). In Kelly's book, the setting is more
restrictive, but the proof carries over
basically unchanged. 

\begin{mythm}[Kelly]\label{thm:Kelly}
Consider a Markovian Petri net.  
Assume that the NLTE (\ref{eq:NLTE}) admit a strictly positive solution
$(u_p)_{p \in \mathcal P}$. Then the marking process of the Petri net
has an invariant measure $\pi$ defined by: $\forall x\in \cR(M_0)$, 
\begin{equation}
\pi(x) = \Phi(x)^{-1} \prod_{p \in \mathcal P} u_p^{x_p} \,. \label{eq:thmKelly}
\end{equation}
\end{mythm}

We then say that $\pi$ has a {\em product form}: $\pi(x)$ decomposes as a
product over the places $p$ of terms depending only on the local
marking $x_p$. 

Observe that $\pi(x)>0$ for all $x$ in \eref{eq:thmKelly}. 
In particular it implies that the marking process is
irreducible.
On the other hand, the measure defined in \eref{eq:thmKelly} may have a finite or
infinite mass. When it has a finite mass, the marking process is
ergodic, and the normalization of $\pi$ is the unique stationary
distribution. 

In the case of mass-action rates \eref{eq-massaction}, we get
\[
\sum_{x\in \cR(M_0)} \pi(x) = \sum_{x\in \cR(M_0)} \prod_{p \in
  \mathcal P}\frac{u_p^{x_p}}{x_p !} \ \leqslant \ \sum_{x\in \N^{\cP}} \prod_{p \in
  \mathcal P}\frac{u_p^{x_p}}{x_p !} = \ \exp (\sum_p u_p) < +\infty \:.
\]
So we are always
in the ergodic case. For constant rates \eref{eq-constant}, if the
state space $\cR(M_0)$ is infinite, the
ergodicity depends on the values of the constants $\kappa_t$. 






\medskip

Theorem \ref{thm:Kelly} is the core result. Below, all the
developments consist in determining conditions under which Theorem
\ref{thm:Kelly} applies. More precisely, we want 
conditions on the model ensuring the existence
of a strictly positive solution to the NLTE and the finiteness of the measure
$\pi$. The ideal situation is as follows:
\begin{itemize}
\item {\em structural} properties of the Petri net (i.e. independent of the
firing rates) ensure the existence of a strictly positive solution to the NLTE; 
\item conditions on the firing rates ensure the finiteness of the
  measure $\pi$. 
\end{itemize}

\subsection{Linear traffic equations and Haddad $\&$ al's Theorem}

Solving the non-linear traffic equations is still a challenging task.
We may avoid a direct attack to these equations by considering a 
simpler system of equations called the linear traffic equations.

\begin{mydef}[Linear traffic equations]\label{def:LTE}
We call {\em linear traffic equations (LTE)} the equations over the
unknowns $(y_C)_{C\in \cC}$ defined by: $\forall C \in \cC$, 
\begin{equation}
y_C \sum_{t: I(t) = C} \kappa_t  = \sum_{t: O(t) = C} \kappa_{t}
y_{I(t)}\,.
\label{eq:LTE}
\end{equation}
Furthermore, if $\emptyset \in \cC$, then $y_{\emptyset} = 1$. 
\end{mydef}

The NLTE and the LTE are clearly linked. 

\begin{mylem}\label{le:nlte-lte}
If the NLTE (\ref{eq:NLTE}) have a strictly positive solution
$u=(u_p)_{p\in \cP}$, then $v=(v_C)_{C\in \cC}$, 
\begin{equation}\label{eq-le}
v_C = \prod_{p: C_p \neq 0} u_p^{C_p} \:,
\end{equation}
is a strictly positive solution to the LTE (\ref{eq:LTE}).
\end{mylem}

For a partial converse statement, see Lemma \ref{le:LTE-NLTE}. 
The following proposition provides a simple and structural criterium for the existence of a
strictly positive solution to the LTE. 

\begin{myprop}\label{prop:wrandLTE}
The following statements are equivalent:
\begin{itemize}
\item $\exists (\kappa_t)_{t \in \cT}$ such that the equations
(\ref{eq:LTE}) have a strictly positive solution.
\item $\forall (\kappa_t)_{t \in \cT}$, the equations
(\ref{eq:LTE}) have a strictly positive solution.
\item The Petri net is weakly reversible.
\end{itemize}

\end{myprop}

Proofs can be found in \cite[Theorem 3.5]{BoSe} or \cite[Corollary
  4.2]{feinberg79}. 
We recall the
argument from \cite{BoSe} which is simple and illuminating. 

 \begin{proof}
 The {\em reaction process} is a continuous-time Markov process, analog to the
 marking process, except that it is built on the reaction
 graph instead of the marking graph. More precisely, the state space is
 the set of complexes $\cC$ and the infinitesimal generator $\widetilde{Q}=
 (\widetilde{q}_{u,v})_{u,v}$
 is defined by:
 \[
 \forall u \neq v, \quad \widetilde{q}_{u,v} = 
 \sum_{t: I(t)=u, O(t)=v} \kappa_t \:. 
 \]
 (The discrete-time version of this process was introduced in
 \cite{HLTa} under the name ``routing process''.)
 The key observation is that the LTE \eref{eq:LTE} are precisely the balance equations $y\widetilde{Q}=0$ of the reaction
 process. The result now follows using standard Perron-Frobenius theory. 
 \end{proof}

So weak reversibility is a necessary condition to have a strictly positive
solution to the NLTE, and to be able to apply
Theorem \ref{thm:Kelly}. Unfortunately, it is not a sufficient
condition as shown by the following example.

\begin{myeg}
Let us consider a Markovian Petri net whose underlying Petri net is
shown in Figure \ref{fig:LTEnotNLTE}, and is 
equivalently defined by the chemical reactions \eref{eq-chem}. 
This is a weakly reversible Petri net, thus its LTE always
have a
strictly positive solution regardless of the choice of the constants $\kappa_t$. The NLTE are:
\begin{equation}\label{eq:LTEnotNLTE}
\kappa_1 x_1  =  \kappa_2 x_2 \qquad \kappa_3 x_3  =   \kappa_4 x_4
\qquad \kappa_5 x_1 x_3  = \kappa_6 x_2 x_4 \:.
\end{equation}
The system (\ref{eq:LTEnotNLTE}) does not always have a strictly positive solution. For
example, set $\kappa_1 = \kappa_2 = \kappa_3 = \kappa_4 = \kappa_5 =
1$, and $\kappa_6 = 2$. Any solution to (\ref{eq:LTEnotNLTE}) must
satisfy either $x_1=x_2=0$ or $x_3=x_4=0$. 

Depending on the values of the constants $(\kappa_t)_t$, the Markovian
Petri net may
or may not have a product form invariant measure. Anticipating on
Theorem \ref{thm:defzero}, the deficiency of the Petri net has to be
different from 0 for this property to hold. 
And, indeed, we have $\mbox{rank}(A)=3$ and $\mbox{rank}(N)=2$ so the
deficiency is equal to 1. 

\end{myeg}

So now the goal is to find additional conditions on top of weak
reversibility to ensure the existence of a product form. 

\medskip

An early result in this direction appears in
Coleman, Henderson and Taylor~\cite[Theorem 3.1]{CHTa96}. The
condition is not structural (i.e. 
rate dependent) and not very tractable. The next result, due to 
Haddad, Moreaux, Sereno, and Silva~\cite[Theorem 9]{HMSS}, provides a
structural sufficient condition. 

\begin{myprop}\label{pr:HMSS}
Consider a Markovian Petri
net (set of complexes $\cC$). Assume
that the Petri net is weakly reversible. 
Let $N$ be the incidence
matrix of the Petri net, see \eref{eq-def-incidence}. Let $A$ be the node-arc incidence matrix of
the reaction graph, see \eref{eq-A}. 
Assume that there exists a $\Q$-valued $(\cC\times \cP)$-matrix $B$
such that $BN=A$. Then the marking process has an invariant measure
$\pi$ given by: $\forall x \in \cR(M_0)$,
\[
\pi(x) = \Phi(x)^{-1} \prod_{p\in \cP} \bigl( \prod_{C\in \cC}
v_C^{B_{C,p}} \bigr)^{x_p} \:,
\]
where $v$ is a strictly positive solution to the LTE. 
\end{myprop}

\subsection{Deficiency zero and product form}\label{sse-def0}

Independently of the efforts in the Petri net community
(\cite{CHTa96,HMSS}), the following result was proved on the
chemical side by Feinberg~\cite[Theorem 5.1]{feinberg79}. 

\begin{mythm}[Feinberg]\label{thm:Feinberg}
Consider a Markovian Petri net. Assume that the Petri net has deficiency 0. 
Then the NLTE have a strictly positive solution if and only if the network is weakly reversible.
\end{mythm}

By combining Theorems \ref{thm:Kelly} and \ref{thm:Feinberg}, we
obtain the following result whose formulation is original. 

\begin{mythm}\label{thm:defzero}
Consider a Petri net which is weakly 
reversible and has deficiency 0.
Consider any associated Markovian Petri net. The NLTE
have a strictly positive solution $(u_p)_p$ and the marking
process has a product form invariant measure:
\begin{equation}\label{eq:thmdefzero}
\pi(x) = \Phi(x)^{-1} \prod_{p \in \mathcal P} u_p^{x_p}\,.
\end{equation}
If we assume furthermore that the rates are of mass-action type
\eref{eq-massaction}, then the marking process is ergodic and its
stationary distribution is:
\[
\pi(x) = C \prod_{p \in \mathcal P} \frac{u_p^{x_p}}{x_p !} \:,
\]
where $C =\bigl( \sum_{x} u_p^{x_p}/x_p ! \bigr)^{-1}$.
\end{mythm}

The above result is interesting. Indeed, the ``deficiency 0'' condition is
structural and very simple to handle. We now prove that the result in
Theorem \ref{thm:defzero} is
equivalent to the one in Proposition \ref{pr:HMSS}. 

\begin{myprop}\label{pr:HMSS2}
Consider a Petri net. There exists a $(\cC\times
\cP)$-matrix $B$
such that $BN=A$ (with the notations of Prop. \ref{pr:HMSS}) if and
only if the 
Petri net has deficiency 0. 
\end{myprop}

\begin{proof}
The deficiency of the Petri net is 0 iff $\mbox{rank}(N) =
\mbox{rank}(A)$. 

Assume first that there exists a matrix $B$ such that $BN=A$.
Since $BN=A$, we have 
$\mbox{rank}(A) \leqslant \mbox{rank}(N)$. By Proposition
\ref{pr-nonneg}, we also have $\mbox{rank}(A) \geqslant
\mbox{rank}(N)$. Therefore the Petri net has deficiency 0.

\medskip

We now prove the converse result. Assume that the Petri net has
deficiency 0. 
Set $r = |\cC| - \ell$. We have
$\mbox{rank}(A) = \mbox{rank}(N) = r$.
Since $\mbox{rank}(A) = r$, we know that there exists an invertible and $\Q$-valued $(\cT
\times \cT)$-matrix $Q$ such that the first $r$ column vectors of $AQ$ are
linearly independant and the last $(|\cT| - r)$ column vectors are
$(0,\dots, 0)^T$.
According to (\ref{eq:eliminationAN2}), the last $(|\cT| - r)$ column
vectors of $NQ$ are $(0,\dots, 0)^T$. But $\mbox{rank}(N) = r$, so the first $r$
column vectors of $NQ$ must be linearly independant.
Denote by $AQ_1, \dots, AQ_r$, {\em resp.} $NQ_1, \dots, NQ_r$,  the first $r$ column
vectors of $AQ$, {\em resp.} $NQ$. Since the two families are
independent, we know that there exists 
a $\Q$-valued $(\cC \times \cP)$-matrix $B$ such that $BNQ_i =
AQ_i,$ for all $i = 1, \dots, r$. In other words,
\begin{equation}\label{eq:BNQ}
BNQ = AQ\,.
\end{equation}
Finally, right-multiplying both sides of \eref{eq:BNQ} by $Q^{-1}$, we
obtain $BN = A$.
\end{proof}

\paragraph{Roadmaps.}

The main results are summarized in the diagram of Figure
\ref{fi-recap}. In the diagram, (N)LTE is a shorthand for: ``the
(N)LTE have a strictly positive solution''. 

\begin{figure}[h]
\[ \epsfxsize=300pt \epsfbox{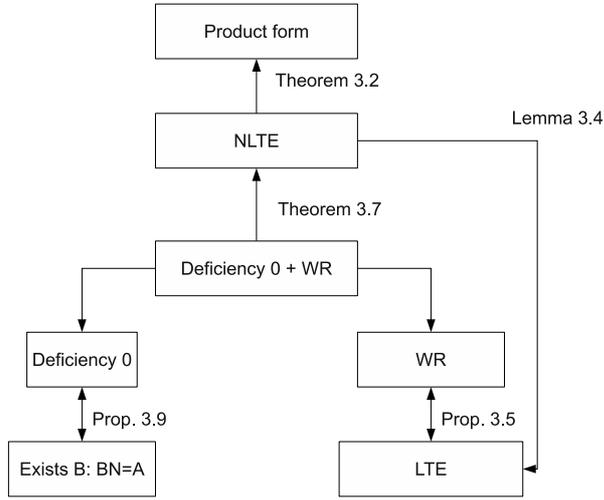} \]
\caption{Roadmap of the results}
\label{fi-recap}
\end{figure}

The implications which do not appear on Figure
\ref{fi-recap} are false. The counter-examples are summarized on 
 Figure \ref{fi-recap2}. (Using Lemma \ref{le:nlte-lte} and
 Prop. \ref{prop:wrandLTE}, any Petri net 
 of deficiency 0 and not WR is a counter-example to [``deficiency 0'' $\implies$
   ``NLTE for some rates''].)

\begin{figure}[H]
\[ \epsfxsize=300pt \epsfbox{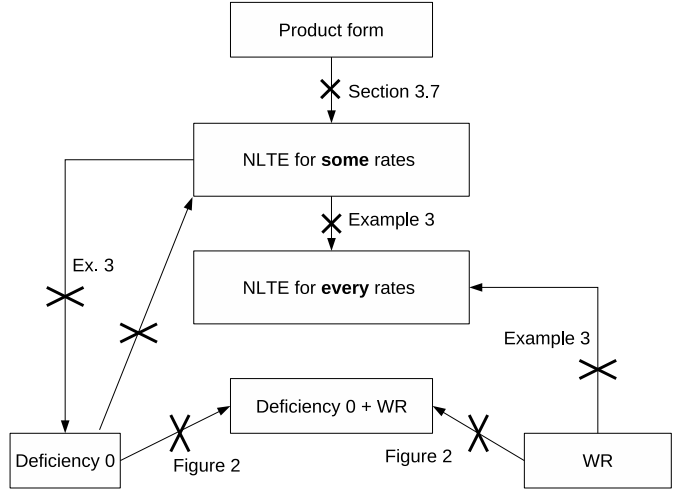} \]
\caption{Roadmap of the counter-examples}
\label{fi-recap2}
\end{figure}

The
product form property depends  on the initial marking via the
reachability set, while the other properties appearing in Figure
\ref{fi-recap} do not depend on it
by definition. This provides a striking picture. 
In fact, for Petri nets which are weakly reversible and have
deficiency zero, when the initial marking varies, the product form
remains the ``same'' but is defined on different state spaces. 
We illustrate this point 
in \S \ref{sse-haddad}, and we come back to it in \S \ref{sse-cex}.

\subsection{A detailed example}\label{sse-haddad}

Consider the Petri graph represented on the left of Figure \ref{fi-haddad}. 
The corresponding reaction graph is given on the right of the figure.

The reaction graph is strongly connected so the
Petri graph is weakly reversible. 
The incidence matrices $N$ and $A$ are given by (indices are ranked as
$(p,q,r)$, $(t_1,t_2,t_3)$, and $(2p,p+q+r,2q)$):
\begin{equation*}
N = \left(
\begin{array}{ccc}
-1 & 2 & -1  \\
1 & -2 & 1  \\
1 & 0 & -1  
\end{array}
\right)\,, \qquad 
A= \left(
\begin{array}{ccc}
-1 & 1 & 0  \\
1 & 0 & -1  \\
0 & -1 & 1 
\end{array}
\right)\:.
\end{equation*}

We check that $\mbox{rank}(A)=\mbox{rank}(N)=2$, so the deficiency is
0. We are in the scope of application of the results of Section
\ref{sse-def0}. 

Denote a marking $M$ by the triple $(M_p,M_q,M_r)$. 
Consider the two Petri nets corresponding to the above Petri graph
with two different initial markings: $(2,0,0)$ and
$(3,0,0)$. 

\begin{figure}[H]
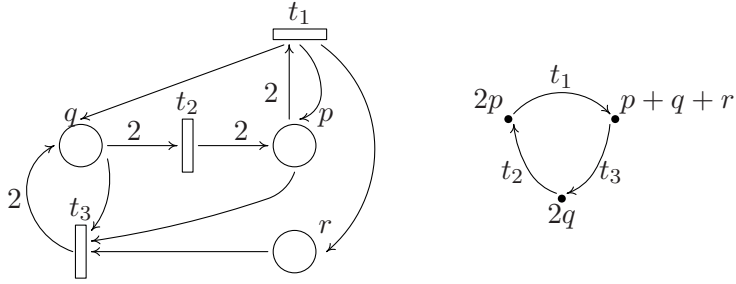

\centering
\begin{mfpic}{-90}{90}{-90}{90}
\circle{(40, 0), 8}
\circle{(-40, 0), 8}
\circle{(40, -40), 8}

\rect{(-2, 10), (2, -10)}
\rect{(-42, -30), (-38, -50)}
\rect{(32, 44), (52, 40)}

\arrow\polyline{(-30, 0), (-4, 0)}
\arrow\polyline{(4, 0), (30, 0)}
\arrow\polyline{(38, 10), (38, 38)}
\arrow\polyline{(30, -40), (-36, -40)}

\arrow\curve{(-43, -40), (-60, -20), (-50, 0)}
\arrow\curve{(-30, -5), (-30, -20), (-36, -32)}
\arrow\curve{(40, -10), (30, -20), (-36, -36)}
\arrow\polyline{(36, 38), (-40, 10)}
\arrow\curve{(42, 38), (50, 20), (42, 10)}
\arrow\curve{(50, 38), (70, 0), (52, -40)}

\tlabelsep{2pt}
\tlabel[bl](47, 5){$p$}
\tlabel[br](-39, 5){$q$}
\tlabel[bl](47, -35){$r$}

\tlabel[bc](0, 10){$t_2$}
\tlabel[bc](-40, -30){$t_3$}
\tlabel[bc](42, 44){$t_1$}

\tlabel[bc](-20, 0){$2$}
\tlabel[bc](20, 0){$2$}
\tlabel[cr](36, 20){$2$}
\tlabel[cr](-60, -20){$2$}

\point[3pt]{(120, 10), (160, 10), (140, -20)}
\tlabel[br](120, 10){$2p$}
\tlabel[bl](160, 10){$p+q+r$}
\tlabel[tc](140, -20){$2q$}

\arrow\curve{(122, 12), (140, 20), (158, 12)}
\arrow\curve{(158, 8), (152, -10), (142, -18)}
\arrow\curve{(138, -18), (128, -10), (122, 8)}

\tlabel[bc](140, 20){$t_1$}
\tlabel[cr](128, -10){$t_2$}
\tlabel[cl](152, -10){$t_3$}

\end{mfpic}
\caption{A Petri graph and its reaction graph.}\label{fi-haddad}
\end{figure}

The two Petri nets have a drastically different behaviour. The first
one is live and bounded, while the second one if live and unbounded. 
The sets of reachable markings are, respectively,
\begin{eqnarray*}
\cR(2,0,0) & = &   \bigl\{ (2,0,0), (1,1,1),(0,2,0) \bigr\} \\
\cR(3,0,0) & = &   \bigl\{ (1,2,0), (0,3,1) \bigr\} \cup \bigl\{
(i,3-i, 2n+1-i), \ 0\leq i \leq 3 , n\geq 1  \bigr\} 
\end{eqnarray*}

For the initial marking $(2,0,0)$, the marking graph is the elementary
circuit $(2,0,0)\longrightarrow (1,1,1) \longrightarrow (0,2,0)
\longrightarrow (2,0,0)$. 
For the initial marking $(3,0,0)$,
the marking graph is represented in Figure
\ref{fi-marking}. The dashed arrows correspond to transition $t_1$,
the dash-and-dotted ones to $t_2$, and the plain ones to $t_3$. 

\begin{figure}[H]
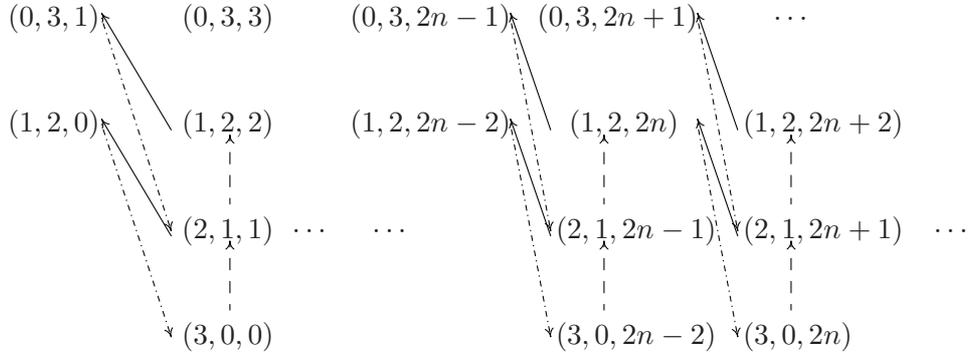

\centering
\begin{mfpic}{-300}{100}{-100}{100}
\tlabelsep{2pt}

\tlabel[cl](-285, 80){$(0, 3, 1)$}
\tlabel[cl](-285, 40){$(1, 2, 0)$}

\tlabel[cl](-220, 80){$(0, 3, 3)$}
\tlabel[cl](-220, 40){$(1, 2, 2)$}
\tlabel[cl](-220, 0){$(2, 1, 1)$}
\tlabel[cl](-220, -40){$(3, 0, 0)$}

\tlabel[cl](-157, 80){$(0, 3, 2n - 1)$}
\tlabel[cl](-157, 40){$(1, 2, 2n - 2)$}
\tlabel[cr](-130, 0){$\dots$}

\tlabel[cl](-87, 80){$(0, 3, 2n + 1)$}
\tlabel[cl](-75, 40){$(1, 2, 2n)$}
\tlabel[cl](-80, 0){$(2, 1, 2n - 1)$}
\tlabel[cl](-80, -40){$(3, 0, 2n - 2)$}

\tlabel[cr](20, 80){$\dots$}
\tlabel[cl](-10, 40){$(1, 2, 2n + 2)$}
\tlabel[cl](-10, 0){$(2, 1, 2n + 1)$}
\tlabel[cl](-10, -40){$(3, 0, 2n)$}

\dashpattern{ddashed}{0pt, 2pt, 2pt, 2pt}

\arrow\gendashed{ddashed}\polyline{(-248, 80), (-222, 0)}
\arrow\gendashed{ddashed}\polyline{(-248, 40), (-222, -40)}

\arrow\dashed\polyline{(-200, -30), (-200, -4)}
\arrow\dashed\polyline{(-200, 10), (-200, 36)}

\arrow\polyline{(-222, -2), (-248, 42)}
\arrow\polyline{(-222, 38), (-248, 82)}

\arrow\gendashed{ddashed}\polyline{(-95, 80), (-80, 0)}
\arrow\gendashed{ddashed}\polyline{(-95, 40), (-80, -40)}

\arrow\dashed\polyline{(-60, -30), (-60, -4)}
\arrow\dashed\polyline{(-60, 10), (-60, 36)}

\arrow\polyline{(-80, -2), (-95, 42)}
\arrow\polyline{(-80, 38), (-95, 82)}

\arrow\gendashed{ddashed}\polyline{(-25, 80), (-10, 0)}
\arrow\gendashed{ddashed}\polyline{(-25, 40), (-10, -40)}

\arrow\dashed\polyline{(10, -30), (10, -4)}
\arrow\dashed\polyline{(10, 10), (10, 36)}

\arrow\polyline{(-10, -2), (-25, 42)}
\arrow\polyline{(-10, 38), (-25, 82)}

\tlabel[cr](-160, 0){$\dots$}
\tlabel[cr](80, 0){$\dots$}

\end{mfpic}
\caption{Marking graph with the initial marking  $(3, 0, 0)$.}\label{fi-marking}
\end{figure}

Consider the Markovian Petri nets associated with the above Petri nets
and constant firing rates $(\kappa_1,\kappa_2,\kappa_3)$ for
$(t_1,t_2,t_3)$.  The NLTE over the unknowns $(x_p,x_q,x_r)$ are given by
\[
\kappa_1 x_p^2 = \kappa_2 x_q^2, \quad 
\kappa_3 x_px_qx_r = \kappa_1 x_p^2, \quad
\kappa_2 x_q^2 = \kappa_3 x_px_qx_r \:.
\]

A strictly positive solution to the NLTE is
$( \sqrt{\kappa_2}/\sqrt{\kappa_1} , \ 1 \ ,
\sqrt{\kappa_1\kappa_2}/\kappa_3  )$. 
Let $\cR$ denote the set of reachable markings. According to Theorem
\ref{thm:defzero}, the
invariant measure $\pi$ is given by 
\begin{equation}\label{eq-invameas}
\forall m=(m_p,m_q,m_r) \in \cR, \quad \pi(m) = \kappa_1^{(m_r-m_p)/2}
\kappa_2^{(m_p+m_r)/2}\kappa_3^{-m_r} \:.
\end{equation}
The invariant measure is expressed in  exactly the same way for the
two Petri nets. 
But it corresponds to two very different situations. 

\medskip

For the initial marking $(2,0,0)$, we have $|\cR|=3$, the model is
ergodic and the unique stationary distribution $p$, obtained by
normalization of \eref{eq-invameas}, is given by 
\[
p(2,0,0) = C^{-1}\kappa_2\kappa_3, \quad p(1,1,1)=
C^{-1}\kappa_1\kappa_2, \quad p(0,2,0) = C^{-1} \kappa_1\kappa_3\:,
\]
with $C= \kappa_1\kappa_2 + \kappa_2\kappa_3+ \kappa_1\kappa_3$. 

\medskip

For the initial marking $(3,0,0)$, we have $|\cR|=\infty$, and the
model is ergodic if and only if the following stability condition is satisfied
\[
\kappa_1\kappa_2 < \kappa_3^2 \:.
\]

It is interesting to observe that we get a non-linear stability
condition. 

\subsection{Additional results}\label{sse-addi}

The statements of Proposition \ref{pr:HMSS} and
Theorem \ref{thm:defzero} differ in that they rely respectively
on the LTE and the NLTE. In particular, it seems at first glance that
Prop. \ref{pr:HMSS} manages to bypass the NLTE. But it is not the
case: the NLTE are hidden in the matrix $B$, see below. 

\begin{mylem}\label{le:LTE-NLTE}
Consider a weakly reversible Markovian Petri net. Assume that there exists a $\Q$-valued matrix $B$
such that $BN=A$ (notations of Prop. \ref{pr:HMSS}). Let
$v=(v_C)_{C\in \cC}$ be a strictly positive solution to the LTE. 
Then $u = (u_p)_{p \in \cP}, u_p
= \prod_{C\in \cC} v_C^{B_{C,p}}$ is a strictly positive solution to the NLTE.
\end{mylem}

\begin{proof}
For each transition $t$, we have:
\begin{eqnarray*}
\prod_p u_p^{I(t)_p - O(t)_p} & = &  \prod_p u_p^{-N_{p,t}}  \qquad = \ \prod_p
\bigl( \prod_C v_C^{B_{C,p}} \bigr)^{-N_{p,t}} \\
& = &  \prod_C v_C^{ -\sum_p B_{C,p} N{p,t} } \ = \ \prod_C v_C^{
  -A_{C,t} } \ = \ \frac{v_{I(t)}}{v_{O(t)}}\:. 
\end{eqnarray*}

Then we have 
\begin{eqnarray*}
\bigl[ u\ \mbox{strictly positive sol. NLTE} \bigr] \ 
&\Longleftrightarrow& \ \forall C, \quad \sum_{I(t) = C} \kappa_t = \sum_{O(t) = C} \kappa_t
\prod_p u_p^{I(t)_p - O(t)_p} \\
&\Longleftrightarrow& \ \forall C, \quad  \sum_{I(t) = C} \kappa_t = \sum_{O(t) = C} \kappa_t
\frac{v_{I(t)}}{v_{O(t)}} \\
&\Longleftrightarrow& \ \bigl[ v\ \mbox{strictly positive sol. LTE} 
  \bigr] \:.
\end{eqnarray*}
\end{proof}

Let us comment on a specific point. Consider Theorem
\ref{thm:defzero}. The invariant measure is defined in fonction 
of a specific strictly positive solution to the NLTE. However, it is
easily seen that the NLTE may have several strictly positive
solutions. Is this
contradictory with the uniqueness of the stationary measure in the
ergodic case~? In the non-ergodic case, do we get several invariant
measures~? The next result answers these questions. 

\begin{mylem}
Assume that the Petri net is weakly reversible and has deficiency 0. Let
$u$, $\widetilde u$ be solutions to the NLTE and $\pi$, $\widetilde \pi$ be
the corresponding invariant measures (given by \eref{eq:thmdefzero}). Then there
exists a constant $K$ such that for all reachable marking $x$,
$\widetilde \pi(x) = K \pi(x)$. 
\end{mylem}

\begin{proof}
It suffices to show that
\begin{equation}\label{eq:pi-pitilde}
\frac{\pi(x - I(t) + O(t))}{\pi(x)} = \frac{\widetilde\pi(x - I(t) +
O(t))}{\widetilde\pi(x)}\,,
\end{equation}
for all reachable marking $x$ and for all enabled transition $t$ of $x$.
Define $v=(v_C)_{C\in \cC}$, and $\widetilde v=(\widetilde v_C)_{C\in
  \cC}$ by
\[ v_C = \prod_{p} u_p^{C_p} ,\qquad  \widetilde v_C = \prod_{p} \widetilde
u_p^{C_p} \:.\]
According to Lemma \ref{le:nlte-lte}, $v$ and $\widetilde v$ are solutions to
the LTE.
Equality \eref{eq:pi-pitilde} holds if and only if 
\begin{equation*}
\prod_p u_p^{O(t)_p - I(t)_p} =
\prod_p
\widetilde u_p^{O(t)_p - I(t)_p}  
\Longleftrightarrow \frac{v_{O(t)}}{v_{I(t)}} = \frac{\widetilde
v_{O(t)}}{\widetilde v_{I(t)}} 
\Longleftrightarrow \frac{\widetilde v_{O(t)}}{v_{O(t)}} =
\frac{\widetilde v_{I(t)}}{v_{I(t)}}\,. 
\end{equation*}
The last equality is proved by Feinberg in \cite[Proposition 4.1]{feinberg79}.
\end{proof}

\subsection{Algorithmic complexity} 

Let us compare Proposition \ref{pr:HMSS} and
Theorem \ref{thm:defzero} from an algorithmic
point of view. 

In both cases, one needs to check weak
reversibility. Using Proposition \ref{prop:wrandLTE}, weak reversibility is
equivalent to the existence of a strictly positive solution to the
LTE. This last point can be checked in time $\mathcal O ( \cC^3)$.  
Then the procedures diverge. 

\begin{itemize}

\item {\em Proposition \ref{pr:HMSS}.} 
We need to compute the matrix $B$ satisfying $BN
= A$. This requires to solve $\cC$ linear systems of dimension $\cP\times
\cT$. So the time-complexity is  $\mathcal O ( \cC \cP\cT^2)$ using
Gaussian elimination.

\item {\em Theorem \ref{thm:defzero}.} 
We have seen in Section \ref{subsec:chem} that the deficiency 0 condition can be
checked in time $\mathcal O ( \cP \cT^2)$. Then one needs to compute a strictly positive solution to
the NLTE. This can be done as follows. Consider \eref{eq-le}, apply
the logarithm operation on both sides and solve the linear system. The
corresponding time-complexity is $\mathcal O ( \cP \cC^2) $ (= $\mathcal O ( \cP
\cT^2)$ since $|\cC|\leq 2|\cT|$). 

\end{itemize}

We conclude that it is more efficient asymptotically to determine the
product form by 
using the characterization in Theorem \ref{thm:defzero}.

\subsection{A product form Petri net with no solution to the NLTE}\label{sse-cex}

We have seen above that the existence of a strictly positive solution
to the NLTE is a sufficient condition for the existence of a product
form. However, it is not a necessary condition, and we now provide a
counter-example. 

\medskip

\begin{figure}[H]
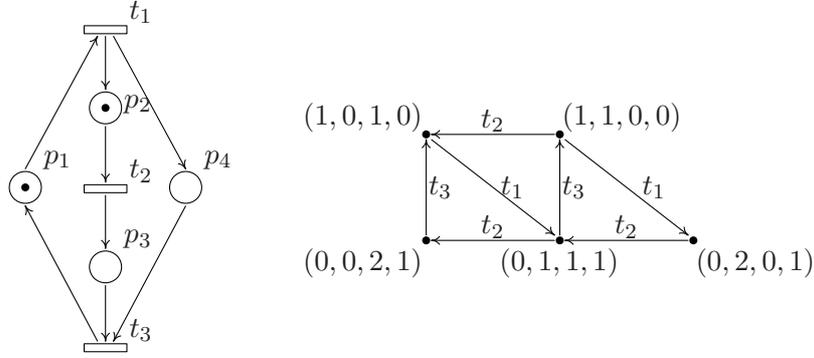

\centering

\begin{mfpic}{-80}{80}{-65}{65}
\tlabelsep{1pt}

\tlabel[tl](-94, 36){$p_2$}
\shiftpath{(-100, 0)}\circle{(0, 30), 6}
\point[3pt]{(-100, 30)}
\tlabel[bl](-94, -24){$p_3$}
\shiftpath{(-100, 0)}\circle{(0, -30), 6}
\tlabel[bl](-124, 6){$p_1$}
\shiftpath{(-100, 0)}\circle{(-30, 0), 6}
\point[3pt]{(-130, 0)}
\tlabel[bl](-64, 6){$p_4$}
\shiftpath{(-100, 0)}\circle{(30, 0), 6}

\shiftpath{(-100, 0)}\rect{(-8, 61), (8, 58)}
\tlabel[bl](-92, 61){$t_1$}
\shiftpath{(-100, 0)}\rect{(-8, 1), (8, -2)}
\tlabel[bl](-92, 1){$t_2$}
\shiftpath{(-100, 0)}\rect{(-8, -59), (8, -62)}
\tlabel[bl](-92, -59){$t_3$}

\arrow\shiftpath{(-100, 0)}\polyline{(-30, 7), (-3, 57)}
\arrow\shiftpath{(-100, 0)}\polyline{(30, -7), (3, -58)}
\arrow\shiftpath{(-100, 0)}\polyline{(3, 57), (30, 7)}
\arrow\shiftpath{(-100, 0)}\polyline{(-3, -58), (-30, -7)}
\arrow\shiftpath{(-100, 0)}\polyline{(0, 57), (0, 37)}
\arrow\shiftpath{(-100, 0)}\polyline{(0, 23), (0, 2)}
\arrow\shiftpath{(-100, 0)}\polyline{(0, -3), (0, -23)}
\arrow\shiftpath{(-100, 0)}\polyline{(0, -37), (0, -58)}

\point[3pt]{(20, 20), (70, 20), (20, -20), (70, -20), (120, -20)}

\arrow\polyline{(68, 20), (22, 20)}
\arrow\polyline{(72, 18), (118, -18)}
\arrow\polyline{(20, -18), (20, 18)}
\arrow\polyline{(70, -18), (70, 18)}
\arrow\polyline{(22, 18), (68, -18)}
\arrow\polyline{(118, -20), (72, -20)}
\arrow\polyline{(68, -20), (22, -20)}

\tlabel[br](20, 20){$(1, 0, 1, 0)$}
\tlabel[bl](70, 20){$(1, 1, 0, 0)$}
\tlabel[tl](120, -22){$(0, 2, 0, 1)$}
\tlabel[tr](20, -22){$(0, 0, 2, 1)$}
\tlabel[tc](70, -22){$(0, 1, 1, 1)$}

\tlabel[bc](45, 20){$t_2$}
\tlabel[bc](45, -20){$t_2$}
\tlabel[bc](95, -20){$t_2$}
\tlabel[cl](20, 0){$t_3$}
\tlabel[cl](70, 0){$t_3$}
\tlabel[cl](47, 0){$t_1$}
\tlabel[cl](100, 0){$t_1$}

\end{mfpic}

\caption{A product form Petri net which is not weakly reversible and has deficiency 1.}
\label{fig:def1notWR}
\end{figure}


Consider the Petri graph of Figure \ref{fig:def1notWR}.
It is not weakly reversible, so in particular there is no strictly
positive solution to the NLTE. Furthermore, the deficiency is 1.

Consider the initial marking
$M_0 = (1, 1, 0, 0)$. The marking graph is given on the right of
Figure \ref{fig:def1notWR}. By solving directly the balance equations
of the Marking process, we get:
\[
\pi(1010)=\frac{C}{\mu_1\mu_3}, \ \pi(1100)=\frac{C}{\mu_1\mu_2}, \ \pi(0111)=\frac{C}{\mu_2\mu_3},
\pi(0201)=\frac{C}{\mu_2^2}, \pi(0021)=\frac{C}{\mu_3^2} \:.
\]
We have 
$\pi(x, y, z,t) = C
(1/\mu_1)^{x}(1/\mu_2)^{y}(1/\mu_3)^{z}$. So the Petri net has a
product form.

Consider now the same Petri graph with a new initial marking $M_0' =
(1, 1, 1, 1)$. 
The marking graph is still finite (9 states), strongly
connected but we no longer have a product form stationary distribution. Indeed,
suppose that $\pi(x, y, z, t) = C a^xb^yc^zd^t$ where $a, b, c, d$ are strictly
positive constants. The balance equations for $M_1 = (2, 1, 0, 0)$ and
$M_2 = (1, 2, 0, 1)$ are:
\begin{equation*}
  \mu_3 \frac{cd}{a} = \mu_1 + \mu_2\,, \qquad 
  \mu_1 \frac{a}{bd} + \mu_3 \frac{cd}{a} = \mu_1 + \mu_2\,.
\end{equation*}
These equations cannot be simultaneously satisfied, so the stationary
distribution is not of product form. 

\medskip

To summarize, as opposed to the case of ``weakly reversible and
deficiency 0'' Petri nets, the existence of a product form depends on
the initial marking.

\section{Markovian free-choice nets and product
  form}\label{sec:freechoice}

The class of Petri nets whose Markovian versions have a product form is an
interesting one. It is therefore natural to study how this class 
intersects with the classical families of Petri nets: state machines
and free-choice Petri nets.

\medskip

The central result of this section is, in a sense, a negative result. 
We show that within the class of free-choice Petri
nets, the only ones which are weakly reversible are closely related to state
machines. We also show that the Markovian state
machines are ``equivalent to'' Jackson networks. The latter form the
most basic and classical examples of product form queueing networks.

\medskip

From now on, we consider only {\em non-weighted} Petri nets, that is Petri
nets with $I,O: \cT \rightarrow \{0,1\}^{\cP}$. In this case, the
input/output bags can be retrieved from the flow relation $\cF$ and we
can define the Petri net as a quadruple $(\cP,\cT,\cF,M_0)$. 
We also identify complexes with subsets of $\cP$.

\medskip

For a node $x\in \cT\cup \cP$, set $^{\bullet} x = \{y: \ (y,x) \in \cF \}$ and
$x^{\bullet}= \{y: \ (x,y) \in \cF \}$. For a set of nodes $S \subset \cT\cup
\cP$, set $^{\bullet} S = \bigcup_{x \in S}\ ^{\bullet} x$ and
$S^{\bullet}= \bigcup_{x \in S} x^{\bullet}$.

\subsection{State machines} \label{subsec:statemachines}

\begin{mydef}[State machine and generalized state machine]
A non-weighted Petri
net $\cN = (\mathcal P, \mathcal T, \mathcal F, M_0)$ is a:
\begin{itemize}
\item {\em State machine (SM)} if for all transition $t$,
$|^{\bullet}t|  = |t^{\bullet}| = 1$;
\item {\em Generalized state machine (GSM)} if for all transition $t$,
$|^{\bullet}t| \leq 1$, $|t^{\bullet}| \leq 1$.
\end{itemize}
\end{mydef}

\begin{mydef}[Associated state machine]\label{def:associatedsm}
Given a GSM $\cN = (\cP, \cT, \cF, M_0)$, the {\em
associated state machine} is $\cN' = (\cP', \cT, \cF', M_0')$ where:
\begin{itemize}
\item $\cP' = \cP \cup \{p\}, p \notin \cP\,,$
\item $\cF' = \cF \cup \{(p, t), t \in \cT, |^\bullet t| = 0\} \cup \{(t, p), t
\in \cT, |t^\bullet| = 0\}\,,$
\item $\forall x \in \cP, \ M_0'(x) = M_0(x), \qquad M_0'(p) = 0\,.$
\end{itemize}
\end{mydef}

\medskip

The figure below shows, from left to right,  a SM, a GSM and its associated SM.

\begin{figure}[H]
\centering
\begin{mfpic}{-150}{150}{-31}{38}
\tlabelsep{3pt}
\circle{(-6, 0), 6}

\circle{(-72, -8), 6}
\circle{(-72, 28), 6}
\circle{(-120, 10), 6}

\circle{(84, 0), 6}
\circle{(84, 24), 6}
\tlabel[bc](84, 30){$p$}

\rect{(-96, 35.5), (-93, 20.5)}
\rect{(-96, -15.5), (-93, -0.5)}

\rect{(-36, 7.5), (-33, -7.5)}
\rect{(24, 7.5), (21, -7.5)}

\rect{(54, 7.5), (57, -7.5)}
\rect{(114, 7.5), (111, -7.5)}

\arrow\polyline{(-113.25, 11.5), (-96.75, 28)}
\arrow\polyline{(-113.25, 8.5), (-96.75, -8)}

\arrow\polyline{(-92.25, 28), (-78.75, 28)}
\arrow\polyline{(-92.25, -8), (-78.75, -8)}

\arrow\polyline{(-32.25, 0), (-12.75, 0)}
\arrow\polyline{(.75, 0), (20.25, 0)}

\arrow\polyline{(57.75, 0), (77.25, 0)}
\arrow\polyline{(90.75, 0), (110.25, 0)}

\arrow\polyline{(114.75, 0), (120, 0), (120, 24), (90.75, 24)}
\arrow\polyline{(77.25, 24), (48, 24), (48, 0), (53.25, 0)}


\end{mfpic}
\end{figure}

\begin{mylem}\label{lem:isomorphsm}
The reaction graph and the Petri graph of a state machine are
isomorphic. The reaction graph of a GSM and the Petri graph
of its associated SM are isomorphic. 
\end{mylem}

\begin{proof}
In a SM, each complex is just one place. 
Starting from the Petri graph and replacing $[p\rightarrow t \rightarrow q],\
p,q\in \cP, \ t \in \cT,$ by $[p\rightarrow q]$, we get the reaction
graph. For GSM, the mapping is the same with the empty complex
corresponding to the ``new'' place in the associated SM. 
\end{proof}

\begin{mycor}
A SM is weakly reversible iff each 
connected component is strong\-ly connected. A GSM is
weakly reversible iff in the associated SM, each
connected component is strongly connected.
\end{mycor}

In a SM, the complexes are the places. So the NLTE and the LTE
coincide exactly. For a GSM, the complexes are the places
and the empty set. With the convention $y_{\emptyset} = 1$, we
still have that the NLTE and the LTE coincide. The next proposition
follows. 

\begin{myprop}\label{prop:wrstatemachineNLTE}
Consider a weakly reversible GSM. 
For every rates $(\kappa_t)_t$, the
NLTE have a strictly positive solution.
\end{myprop}

\begin{proof}
In the weakly reversible case, the LTE have 
a strictly positive solution for every choice of the rates,
Proposition \ref{prop:wrandLTE}. Therefore the NLTE have a strictly 
positive solution for every choice of the rates.
\end{proof}

The above proof does not require Feinberg's Theorem \ref{thm:Feinberg}.
However, it turns out that the deficiency is 0, which provides a
second proof of Prop. \ref{prop:wrstatemachineNLTE} using Theorem
\ref{thm:Feinberg}. 

\begin{myprop}\label{prop:gensmdef0}
Generalized state machines have deficiency 0.
\end{myprop}

\begin{proof}
Consider first a state machine. Let $N$ be the 
incidence matrix, and let 
$A$ be the node-arc incidence matrix of the reaction graph, see
\eref{eq-A}. 
Using Lemma
\ref{lem:isomorphsm}, we get immediately that $A=N$. So, in particular,
we have $\mbox{rank}(A)=\mbox{rank}(N)$ and the deficiency is 0. 

\medskip

Consider now a GSM $\cN$ and its associated SM
$\cN'$.
Call $\cC$ (resp. $\cC'$), $N$ (resp. $N'$) and $\ell$ (resp. $\ell'$)
the set of complexes, the incidence matrix and the number of connected
components of the reaction graph of $\cN$ (resp. $\cN'$).

Since $\cN$ and $\cN'$ have the same reaction graph
(Lemma
\ref{lem:isomorphsm}), we have:
\begin{equation}
|\cC| = |\cC'|\,,\,\ell = \ell'\,. \label{eq:ccll}
\end{equation}

By construction of $\cN'$, $N'$ is $N$ augmented with a row
$(x_t)_{t \in \cT}$ defined by 
$x_t = {\bf 1}_{\{t^{\bullet} = \emptyset\}} - {\bf
  1}_{\{{}^{\bullet}t=\emptyset\}}$
(where $t^\bullet$ and ${}^\bullet t$ are defined in $\cN$). 
We have $\mbox{rank} (N') \geq \mbox{rank} (N)$. 
On the other hand, observe that $\forall t \in \cT, x_t = - \sum_{s \in \cP}
N_{s, t}$, so $N' = BN$, where $B$ is the $\cP \times \cP$ identity matrix
augmented with the row $(-1, \dots, -1)$. Hence $\mbox{rank} (N') =
\mbox{rank}(BN) \leq \mbox{rank} (N)$.
So:
\begin{equation}
\mbox{rank} (N')
= \mbox{rank} (N)\,. \label{eq:rankrank}
\end{equation}

Together \eref{eq:ccll} and \eref{eq:rankrank} imply that $\cN$ and $\cN'$
have the same deficiency. Since $\cN'$ is a SM, it has deficiency
zero, so $\cN$ also has deficiency zero.
\end{proof}

By coupling Proposition \ref{prop:wrstatemachineNLTE} and Theorem
\ref{thm:Kelly}, or alternatively Proposition
\ref{prop:gensmdef0} and Theorem \ref{thm:defzero}, we get
the result below.

\begin{mythm}\label{th:sm}
Consider a Markovian weakly reversible GSM. The deficiency is 0. 
The NLTE have a strictly positive solution $(u_p)_p$. The marking process admits
a product form invariant measure given by:
$\forall x \in \cR(M_0)$, 
\[
\pi(x) = \Phi(x)^{-1} \prod_{p \in \mathcal P} u_p^{x_p}\,.
\]
If $\emptyset \not\in \cC$, the state space $\cR(M_0)$ is finite, the marking process is
ergodic, and $\pi$ can be normalized to give a product form stationary
distribution:
$\forall x \in \cR(M_0)$, 
\[
\widetilde{\pi}(x) = B \Phi(x)^{-1} \prod_{p \in \mathcal P} u_p^{x_p}\,,
\]
where $B = \bigl( \sum_{x\in \cR(M_0)} \Phi(x)^{-1} \prod_{p \in
  \mathcal P} u_p^{x_p} \bigr)^{-1}$.
\end{mythm}

\medskip

Theorem \ref{th:sm} is far from a surprising or new result, as we now show. 

\subsection{Jackson networks} \label{subsec:Jacksonnets}

The product form result for Jackson networks is one of the
cornerstones of Markovian queueing theory. It was originally proved by
Jackson~\cite{jack57} for open networks and by Gordon $\&$
Newell~\cite{GoNe} for closed networks. 

\medskip

Consider a Markovian weakly reversible SM with constant rates
$(\kappa_t)_{t\in \cT}$.  It can
be transformed into a Jackson network as follows:

\begin{itemize}
\item A place $s$ becomes a simple queue, that is a single server
  Markovian queue with an infinite buffer. The service rate at queue $s$ is
  $\mu_s = \sum_{t\in s^{\bullet}} \kappa_t$. 
\item The routing matrix $P$ of the Jackson network is the stochastic
  matrix defined by: 
\[
\forall u,v \in \cP, \qquad P_{u,v} = \begin{cases}  \mu_u^{-1} \sum_{t: ^{\bullet}t =u, t^{\bullet}
      =v} \kappa_t & \mbox{if } \exists t \in \cT,
  u\rightarrow t \rightarrow v \\
0 & \mbox{otherwise}
\end{cases} \:.
\]
\item A token in place $s$ becomes a customer in queue $s$. 
\end{itemize}

Consider now a Markovian weakly reversible GSM with constant rates
$(\kappa_t)_{t\in \cT}$. On top of the above transformations, we do the
following:

\begin{itemize}
\item A transition $t$ with ${}^{\bullet}t=\emptyset$ becomes an
  external Poisson arrival flow of rate $\kappa_t$ in queue
  $t^{\bullet}$. 
\end{itemize} 

The routing matrix $P$ is now substochastic. Indeed, if the transition $t$
is such that $t^{\bullet}=\emptyset$, then $\sum_v P_{{}^{\bullet}t,v}
< 1$. 

\medskip

In the SM case, the Jackson network is {\em closed}, that is
without arrivals from nor departures to the
outside. In the GSM case with input and output transitions,
the Jackson network is {\em open}. 


A Jackson network can be translated into a Markovian (G)SM
using the reverse construction.

\medskip

The two models are identical in a strong sense. 
Precisely, the marking process of the state machine and the
queue-length process of the Jackson network have
the same infinitesimal generator. 

\begin{figure}[H]
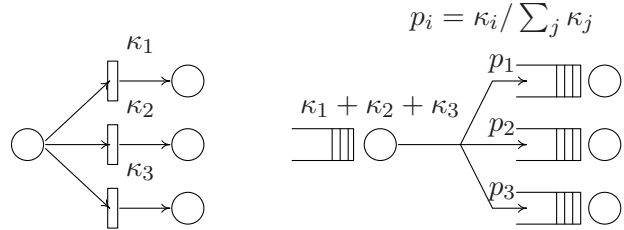

\centering
\begin{mfpic}{-130}{130}{-72}{72}
\tlabelsep{3pt}

\circle{(-114, 0), 6}
\circle{(-54, 0), 6}
\circle{(-54, 24), 6}
\circle{(-54, -24), 6}

\rect{(-84, 7.5), (-80.25, -7.5)}
\tlabel[bl](-80.25, 7.5){$\kappa_2$}

\rect{(-84, 31.5), (-80.25, 16.5)}
\tlabel[bl](-80.25, 31.5){$\kappa_1$}

\rect{(-84, -31.5), (-80.25, -16.5)}
\tlabel[bl](-80.25, -16.5){$\kappa_3$}

\arrow\polyline{(-107.25, 0), (-83.25, 0)}
\arrow\polyline{(-107.25, 1.5), (-83.25, 24)}
\arrow\polyline{(-107.25, -1.5), (-83.25, -24)}

\arrow\polyline{(-79.5, 0), (-60.75, 0)}
\arrow\polyline{(-79.5, 24), (-60.75, 24)}
\arrow\polyline{(-79.5, -24), (-60.75, -24)}


\circle{(18, 0), 6}
\polyline{(-15, 6), (8, 6), (8, -6), (-15, -6)}
\polyline{(6, 6), (6, -6)}
\polyline{(3, 6), (3, -6)}
\polyline{(0, 6), (0, -6)}

\tlabel[bc](18, 7.5){$\kappa_1 + \kappa_2 + \kappa_3$}

\circle{(102, 0), 6}
\polyline{(69, 6), (93, 6), (93, -6), (69, -6)}
\polyline{(87, 6), (87, -6)}
\polyline{(90, 6), (90, -6)}
\polyline{(84, 6), (84, -6)}

\circle{(102, 24), 6}
\polyline{(69, 30), (93, 30), (93, 18), (69, 18)}
\polyline{(87, 30), (87, 18)}
\polyline{(90, 30), (90, 18)}
\polyline{(84, 30), (84, 18)}

\circle{(102, -24), 6}
\polyline{(69, -30), (93, -30), (93, -18), (69, -18)}
\polyline{(87, -30), (87, -18)}
\polyline{(90, -30), (90, -18)}
\polyline{(84, -30), (84, -18)}

\arrow\polyline{(24.75, 0), (72.75, 0)}
\tlabel[bc](63.75, 0){$p_2$}

\arrow\polyline{(48, 0), (60, 24), (72.75, 24)}
\tlabel[bc](63.75, 24){$p_1$}

\arrow\polyline{(48, 0), (60, -24), (72.75, -24)}
\tlabel[bc](63.75, -24){$p_3$}

\tlabel[bc](63.75, 37.5){$p_i = \kappa_i / \sum_j \kappa_j$}


\end{mfpic}
\caption{From (generalized) state machine (left) to Jackson network (right).}
\label{fig:statemachinetoJacksonnet}
\end{figure}

\medskip

The classical product form results for Jackson networks
(Jackson~\cite{jack57} and Gordon $\&$ Newell~\cite{GoNe})
are exactly the translation via the
above transformation of Theorem \ref{th:sm}. In the open case, the
weak-reversibility implies the classical ``without capture'' condition
of Jackson networks. 

\medskip

The above transformation from GSM to queueing network can
also be performed in the case of general rate functions of type
\eref{eq:generalratefunction}. Queueing networks with those rate functions are
called {\em Whittle networks} in the literature. The
existence of product form invariant measures for these networks is a
classical result, see for instance \cite{serf99} and the references
therein. 

\subsection{Free-choice Petri nets}\label{subsec:freechoice}

We study the family of live and bounded free-choice nets. This is an
important class of Petri nets realizing a nice compromise between
modelling power and tractability, see the dedicated
monography of Desel $\&$ Esparza~\cite{DeEs}. 
We show that the only such Petri nets having a product form are, in a
sense, the GSM. 

\begin{mydef}[Free-choice Petri net]\label{def:freechoice}
A {\em free-choice} Petri net is a non-weighted Petri net $(\mathcal P, \mathcal
T, \mathcal F, M_0)$
such that: for every two transitions $t_1$ and $t_2$,
either 
$^{\bullet}t_1  = \ ^{\bullet}t_2$ or $^{\bullet}t_1 \cap \ ^{\bullet}t_2 =
\emptyset$.
\end{mydef}

Some authors call the above an {\em extended free-choice} Petri net 
and have a more restrictive definition for free-choice Petri nets.

\medskip

In the figure below, the Petri net on the left is
free-choice, while the one on the right is not free-choice.

\begin{figure}[H]
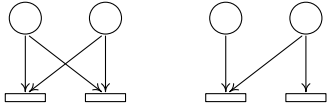

\centering
\begin{mfpic}{-75}{75}{-22.5}{37.5}
\rect{(-60, 1.5), (-45, -1.5)}
\rect{(-30, 1.5), (-15, -1.5)}
\rect{(15, 1.5), (30, -1.5)}
\rect{(45, 1.5), (60, -1.5)}

\circle{(-52.5, 30), 6}
\circle{(-22.5, 30), 6}
\circle{(22.5, 30), 6}
\circle{(52.5, 30), 6}

\arrow \polyline{(-52.5, 23.25), (-52.5, 2.25)}
\arrow \polyline{(-22.5, 23.25), (-22.5, 2.25)}
\arrow \polyline{(-51, 23.25), (-24, 2.25)}
\arrow \polyline{(-24, 23.25), (-51, 2.25)}

\arrow \polyline{(22.5, 23.25), (22.5, 2.25)}
\arrow \polyline{(52.5, 23.25), (52.5, 2.25)}
\arrow \polyline{(51, 23.25), (24, 2.25)}

\tlabelsep{3pt}
\end{mfpic}
\caption{Free-choice (left) and non free-choice (right) Petri nets.}
\label{fig:freechoice}
\end{figure}

\begin{mydef}[Cluster]\label{def:cluster}
The {\em cluster} of a node $x\in \cP\cup \cT$, denoted by $[x]$, is the minimal
set of nodes
such that: $(i)$ $x \in [x]\,;$ $(ii)$ $\forall t \in \mathcal T: \ t \in [x] \implies \ ^{\bullet}t 
\subset [x]\,;$ $(iii)$ $\forall p \in \mathcal P: \ p \in [x] \implies  p^{\bullet} \subset
[x]\,.$
\end{mydef}

The clusters form a partition of the set of nodes, see
\cite[Proposition 4.5]{DeEs}, and
therefore of the places. Moreover, we have the following. 

\begin{mylem}\label{lem:isomorphismclusterscomplexes}
Consider a weakly reversible free-choice Petri net.
The non-empty complexes are disjoint
subsets
of $\cP$. The partition of $\cP$ induced by the non-empty complexes is the same
as the partition of $\cP$ induced by the clusters. 
\end{mylem}

\begin{proof}
In a weakly reversible free-choice Petri net, the non-empty complexes are also
non-empty input bags, which are disjoint according to the definition of
free-choiceness.

It follows from the definition of clusters that every non-empty input bag is
entirely contained in a cluster. This cluster is unique because the clusters
partition the set of places. Let $I$ be
a non-empty complex (which is also a non-empty input bag). Denote by $[I]$ the
cluster containing $I$. We have $I^\bullet \subset [I]$, so $I \cup I^\bullet
\subset [I]$. Since the Petri net is free-choice, ${}^\bullet t = I$
for all $t
\in I^\bullet$. The set $I \cup I^\bullet$ satisfies the three conditions of
the definition of clusters, so we have $[I] \subset I \cup
I^\bullet$. We conclude that $[I] = I \cup
I^\bullet$ and $I$ is the set of places of the cluster $[I]$.

Conversely, let $[x]$ be a cluster such that $[x] \cap \cP \neq \emptyset$.
Let $I$ be a non-empty input bag contained in $[x]$. We have, using
the above, $[I] = I \cup I^\bullet \subset [x]$. By minimality, 
$[I]=[x]$ and $[x]\cap \cP = I$. 
\end{proof}

The above result is not true for a non-WR free-choice Petri net. 
Consider for instance the Petri net represented in the upper-right
part of Figure \ref{fi-4}. 

\medskip

Under the assumptions of Lemma \ref{lem:isomorphismclusterscomplexes},
the non-empty complexes are disjoint. Thus each non-empty complex
behaves as if it was a ``big place''. 
Consider the operation which reduces each non-empty complex to a single
place. The resulting Petri net is a generalized state machine. And this
generalized state
machine is weakly
reversible because the original free-choice Petri net was weakly
reversible. Let us define all this more formally. 

\medskip

\begin{mydef}[Reduced generalized state machine]\label{def:reducedstatemachine}
Let $\cN = (\cP, \cT, \cF, M_0)$ be a weakly reversible free-choice
Petri net with set of complexes $\cC$.
We call the {\em reduced generalized state machine (RGSM)} of $\cN$ the
GSM $\cR\cN =
(\cC \setminus \{\emptyset\},
\cT, \widetilde{\cF}, \widetilde{M}_0)$ where:
\begin{itemize}
\item $\widetilde{\cF} = \{({}^\bullet t, t), (t, t^\bullet), t \in \cT \}$;
\item $\widetilde{M}_0$ is defined by: $\forall C \in \cC \setminus
  \{\emptyset\}$, $\widetilde{M}_0(C) =
\min_{p \in C} M_0(p)$.
\end{itemize}
\end{mydef}

The Petri graph of $\cR\cN$ is the reaction graph of $\cN$
reinterpreted as a Petri graph (the nodes of the reaction graph are
the places, and a transition is added to each arc). 

\begin{mylem}\label{lem:isomorphpnsm}
Let $\cN$ be a weakly reversible free-choice Petri net
and $\cR\cN$ its RGSM. The
marking graph of $\cR\cN$ is isomorphic to the one of $\cN$. If $\cN$ and $\cR\cN$ are
Markovian with the same rates then the two marking processes are ``identical'',
meaning that they have the same infinitesimal generator.
\end{mylem}

\begin{proof}
Consider two places $p$ and $q$ belonging to the same complex. Since
the complexes are disjoint, Lemma
\ref{lem:isomorphismclusterscomplexes}, each time $p$ gains
(resp. loses) a token, so does $q$. So the difference $M_p-M_q$ is
invariant over all the reachable markings $M$. It has the following
consequence. 

Consider $f: \cR(M_0) \rightarrow \N^{\cC\setminus \emptyset}$ 
defined by $f(M)_C = \min_{p \in \cC} M(p)$.
If $M \xrightarrow{t} M'$ in
$\cN$ then $f(M) \xrightarrow{t} f(M')$ in $\cR\cN$. So $f(\cR(M_0)) =
\cR(\widetilde{M}_0)$ and the marking graph of $\cR\cN$ is the marking graph of
$\cN$ up to a renaming of the nodes.

Since the marking graphs are the same, the infinitesimal generators are also
identical if the two Petri nets have the same rates.
\end{proof}

One could introduce the RGSM associated with a non-WR free-choice Petri
net as in Definition \ref{def:reducedstatemachine}, 
see Figure \ref{fig:WRiso}. 
But in this case Lemma \ref{lem:isomorphpnsm} does
not hold, and the two marking graphs may have nothing in common.  

Now let us compare the structural characteristics of the original
free-choice Petri net $\cN$ and of the reduced generalized
state machine $\cR\cN$.

\begin{figure}[H]
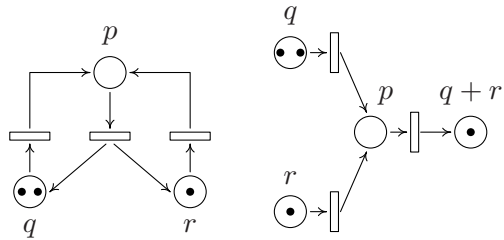

\centering
\begin{mfpic}{-93}{93}{-37.5}{37.5}
\circle{(-52.5, 22.5), 6}
\circle{(-82.5, -22.5), 6}
\circle{(-22.5, -22.5), 6}

\rect{(-45, -3), (-60, 0)}
\rect{(-90, -3), (-75, 0)}
\rect{(-15, -3), (-30, 0)}

\tlabelsep{3pt}

\tlabel[bc](-52.5, 30){$p$}
\tlabel[tc](-82.5, -30){$q$}
\tlabel[tc](-22.5, -30){$r$}

\arrow\polyline{(-52.5, 15), (-52.5, 1.5)}
\arrow\polyline{(-54, -4.5), (-75, -22.5)}
\arrow\polyline{(-51, -4.5), (-30, -22.5)}
\arrow\polyline{(-82.5, -15), (-82.5, -4.5)}
\arrow\polyline{(-22.5, -15), (-22.5, -4.5)}
\arrow\polyline{(-22.5, 1.5), (-22.5, 22.5), (-45, 22.5)}
\arrow\polyline{(-82.5, 1.5), (-82.5, 22.5), (-60, 22.5)}

\point[3pt]{(-85.5, -22.5), (-79.5, -22.5), (-22.5, -22.5)}

\circle{(15, 30), 6}
\circle{(15, -30), 6}
\circle{(45, 0), 6}
\circle{(82.5, 0), 6}

\rect{(30, 37.5), (33, 22.5)}
\rect{(30, -37.5), (33, -22.5)}
\rect{(60, 7.5), (63, -7.5)}

\tlabel[bc](15, 37.5){$q$}
\tlabel[bc](15, -22.5){$r$}
\tlabel[bl](45, 7.5){$p$}
\tlabel[bc](82.5, 7.5){$q + r$}

\point[3pt]{(18.75, 30), (11.25, 30), (15, -30), (82.5, 0)}

\arrow\polyline{(22.5, 30), (28.5, 30)}
\arrow\polyline{(33.75, 30), (43.5, 7.5)}
\arrow\polyline{(22.5, -30), (28.5, -30)}
\arrow\polyline{(33.75, -30), (43.5, -7.5)}

\arrow\polyline{(52.5, 0), (58.5, 0)}
\arrow\polyline{(63.75, 0), (75, 0)}

\end{mfpic}
\caption{A non-WR free-choice net and the associated RGSM.}\label{fig:WRiso}
\end{figure}

\begin{mylem}\label{lem:wrreducedsamedef}
Let $\cN$ be a weakly reversible free-choice Petri net.
The RGSM $\cR\cN$ is weakly
reversible and has the same
deficiency as $\cN$.
\end{mylem}

\begin{proof}
The weak reversibility of $\cR\cN$ follows directly from the definition of
reduced generalized 
state machines.

The Petri graph of $\cR\cN$ is isomorphic to its
reaction graph, Lemma \ref{lem:isomorphsm}. Now by
construction, $\cN$ and $\cR\cN$ have the same reaction graph.
So the number of complexes and the number of connected
components of the reaction graph do not change.
Call $N$ and $N'$ the incidence matrices of $\cN$ and $\cR\cN$ respectively.
Let $C$ be an arbitrary complex, let $p, p'$ be two places of $C$. For every
transition $t$, we have $N_{p, t} = N_{p', t} = {N'}_{C, t}$, which
implies that $\mbox{rank}(N) = \mbox{rank}(N')$.
So the two Petri nets have the same deficiency.
\end{proof}

\begin{mycor}\label{cor:wrfcdef0}
Weakly reversible free-choice Petri nets have deficiency 0.
\end{mycor}



Now all the results for weakly reversible GSM can be applied to
weakly reversible free-choice Petri nets. We get the following.

\begin{mythm}\label{thm:fcwrNLTEdefzero} 
Repeat the statement of Theorem \ref{th:sm} with ``free-choice Petri
net'' replacing ``GSM''. 

\end{mythm}

{\bf Acknowledgement.} We are grateful to Serge Haddad for 
fruitful discussions and for suggesting the example of Section
\ref{sse-haddad}. 


\closegraphsfile

\end{document}